\documentclass[a4paper,12pt]{article}

\usepackage{amsmath,amssymb}
\usepackage{amscd}
\usepackage{mathrsfs}
\usepackage{amsthm}
\usepackage{color}
\usepackage{ascmac}

\usepackage{verbatim}

\usepackage{mathtools}
\usepackage{enumerate}

\usepackage[backref=page, colorlinks=true, linkcolor=blue, citecolor=blue, urlcolor=blue]{hyperref}

\mathtoolsset{showonlyrefs}

\theoremstyle{plain}

\newtheorem{theorem}{\bf Theorem}[section]
\newtheorem{lemma}[theorem]{\bf Lemma}
\newtheorem{proposition}[theorem]{\bf Proposition}
\newtheorem{corollary}[theorem]{\bf Corollary}

\theoremstyle{definition}
\newtheorem{definition}[theorem]{\bf Definition}

\newtheorem{remark}[theorem]{\bf Remark}

  \makeatletter
  \newcommand{\subsubsubsection}{\@startsection{paragraph}{4}{\z@}%
    {1.0\Cvs \@plus.5\Cdp \@minus.2\Cdp}%
    {.1\Cvs \@plus.3\Cdp}%
    {\reset@font\sffamily\normalsize}
  }
  \makeatother
\makeatletter

\@addtoreset{equation}{section}
\makeatother

\title{Classification of abstract Bose field models}

\author{Yasumichi Matsuzawa\thanks{Department of Mathematics, Faculty of Education, Shinshu University, 6-Ro, Nishi-nagano, Nagano 380-8544, Japan, e-mail: myasu@shinshu-u.ac.jp} }
\date{\today}

\begin{document}
\maketitle

\begin{abstract}
	We classify a class of abstract Bose field models in quantum field theory up to unitary equivalence. 
	The class includes abstract free Bose field models, abstract van Hove--Miyatake models, and infrared-renormalized van Hove--Miyatake models. 
	Moreover, as an application of our classification, we classify quadratic interaction models. 
	In particular, our results answer a question of A. Arai concerning the classification of infrared-renormalized van Hove--Miyatake models.
\end{abstract}

\section{Introduction and main results}
We investigate abstract Bose field models in quantum field theory within the Hamiltonian formalism.
An abstract Bose field model is defined by a pair consisting of an irreducible Weyl representation of the canonical commutation relations and a self-adjoint operator, both of which are realized on the same complex Hilbert space.
The Weyl representation describes the time-zero fields of the model, and the self-adjoint operator is the Hamiltonian of the model.
We focus on a class of abstract Bose field models, including abstract free Bose field models, abstract van Hove--Miyatake models, infrared-renormalized van Hove--Miyatake models, and others discussed in \cite[Chapter 10]{MR4292535}.
The main purpose of this paper is to classify the models in this class up to unitary equivalence.
Since unitarily equivalent models are regarded as physically identical, such a classification reveals the essential features of each model.
We establish necessary and sufficient conditions for two models in this class to be unitarily equivalent.
Moreover, as an application of our classification, we classify quadratic interaction models. 
It turns out that the set of quadratic interaction models admits a complete invariant.
Our results, in particular, answer a question of A. Arai concerning the classification of infrared-renormalized van Hove--Miyatake models \cite[Remark 10.19]{MR4292535}.
\newline

\noindent
\textbf{\textit{Abstract Bose field models.}}
To state our results precisely, we first introduce the relevant notation.
We follow the notation and conventions of \cite{MR4292535,MR4812858}.
Let $\mathscr{F}$ be a separable complex Hilbert space, and let $\mathscr{W}$ be a real inner product space.
If a set $\{\phi(f),\pi(f)\mid f\in\mathscr{W}\}$ of self-adjoint operators acting in $\mathscr{F}$ satisfies 
\begin{align*}
	\mathrm{e}^{\mathrm{i}\phi(sf+tg)}&=\mathrm{e}^{\mathrm{i}s\phi(f)}\mathrm{e}^{\mathrm{i}t\phi(g)},
	\qquad \mathrm{e}^{\mathrm{i}\pi(sf+tg)}=\mathrm{e}^{\mathrm{i}s\pi(f)}\mathrm{e}^{\mathrm{i}t\pi(g)},\\
	\mathrm{e}^{\mathrm{i}\phi(f)}\mathrm{e}^{\mathrm{i}\pi(g)}
	&=\mathrm{e}^{-\mathrm{i}\langle f,g\rangle}\mathrm{e}^{\mathrm{i}\pi(g)}\mathrm{e}^{\mathrm{i}\phi(f)},
	\qquad\forall s,t\in\mathbb{R},\ f,g\in\mathscr{W},
\end{align*}
then the pair
\[
\rho:=\{\mathscr{F},\{\phi(f),\pi(f)\mid f\in\mathscr{W}\}\}
\]
is called a \textit{Weyl representation of the canonical commutation relations} over $\mathscr{W}$.
If, in addition, the only bounded operators on $\mathscr{F}$ that commute with all elements of $\{\mathrm{e}^{\mathrm{i}\phi(f)}, \mathrm{e}^{\mathrm{i}\pi(f)}\mid f\in\mathscr{W}\}$ are scalar multiples of the identity, the Weyl representation $\rho$ is called \textit{irreducible}.
For two Weyl representations
\[
\rho_1=\{\mathscr{F}_1,\{\phi_1(f),\pi_1(f)\mid f\in\mathscr{W}\}\}\quad\text{and}\quad \rho_2=\{\mathscr{F}_2,\{\phi_2(f),\pi_2(f)\mid f\in\mathscr{W}\}\}
\]
 of the canonical commutation relations over $\mathscr{W}$, we say that $\rho_1$ is \textit{equivalent} to $\rho_2$ if there exists a unitary operator $U:\mathscr{F}_1\to\mathscr{F}_2$ such that
 \begin{equation}\label{def of equiv of reprs}
 	U\phi_1(f)U^*=\phi_2(f),\qquad U\pi_1(f)U^*=\pi_2(f),\qquad \forall f\in\mathscr{W}.
 \end{equation}

An \textit{abstract Bose field model} is defined as a triple
\[
\mathbb{M}=\{\mathscr{F},H,\{\phi(f),\pi(f)\mid f\in\mathscr{W}\}\},
\]
where $\{\mathscr{F},\{\phi(f),\pi(f)\mid f\in\mathscr{W}\}\}$ is an irreducible Weyl representation of the canonical commutation relations over $\mathscr{W}$ and $H$ is a self-adjoint operator acting in $\mathscr{F}$.
For two abstract Bose field models 
\[
\mathbb{M}_1=\{\mathscr{F}_1,H_1,\{\phi_1(f),\pi_1(f)\mid f\in\mathscr{W}\}\}
\ \ \text{and}\ \ \mathbb{M}_2=\{\mathscr{F}_2,H_2,\{\phi_2(f),\pi_2(f)\mid f\in\mathscr{W}\}\},
\]
we say that $\mathbb{M}_1$ is \textit{equivalent} to $\mathbb{M}_2$ if there exist a unitary operator $U:\mathscr{F}_1\to\mathscr{F}_2$ and a real number $E\in\mathbb{R}$ such that
\begin{equation}
	UH_{1}U^*=H_{2}+E,
\end{equation}
and that \eqref{def of equiv of reprs} holds.
\newline

\noindent
\textbf{\textit{Boson Fock spaces.}}
To define our model concretely, we introduce Boson Fock spaces and the operators that act in them.
Let $\mathscr{H}$ be a separable complex Hilbert space.
Its inner product $\langle f,g\rangle$ is anti-linear in $f\in\mathscr{H}$ and linear in $g\in\mathscr{H}.$
The \textit{Boson Fock space} over $\mathscr{H}$ is defined by
\[
\mathscr{F}_\mathrm{b}(\mathscr{H}):=\bigoplus_{n=0}^\infty\otimes_\mathrm{s}^n\mathscr{H},
\]
where $\otimes_\mathrm{s}^n\mathscr{H}$ is the $n$-fold symmetric tensor product Hilbert space of $\mathscr{H}$ with $\otimes_\mathrm{s}^0\mathscr{H}:=\mathbb{C}$ and $\otimes_\mathrm{s}^1\mathscr{H}:=\mathscr{H}.$
Each vector $\Psi\in\mathscr{F}_\mathrm{b}(\mathscr{H})$ can be expressed as $\Psi=\{\Psi^{(n)}\}_{n=0}^\infty,$ where $\Psi^{(n)}\in\otimes_\mathrm{s}^n\mathscr{H}$ for all $n\geq0.$ 
A subspace
\[
\mathscr{F}_\mathrm{b,0}(\mathscr{H}):=\left\{\Psi=\{\Psi^{(n)}\}_{n=0}^\infty\in\mathscr{F}_\mathrm{b}(\mathscr{H})\mid \exists n_0\in\mathbb{N},\ \forall n\geq n_0,\ \Psi^{(n)}=0\right\}
\]
is dense in $\mathscr{F}_\mathrm{b}(\mathscr{H}).$
A vector 
\[
\Omega_\mathscr{H}:=\{1,0,0,\dots\}\in\mathscr{F}_\mathrm{b}(\mathscr{H})
\]
is called the \textit{Fock vacuum}.

The \textit{creation operator} $A^\dagger(f)$ with test vector $f\in\mathscr{H}$ is defined by
\begin{align*}
	\operatorname{dom}(A^\dagger(f))&:=\left\{\Psi=\{\Psi^{(n)}\}_{n=0}^\infty\in\mathscr{F}_\mathrm{b}(\mathscr{H})\,\left| \,\sum_{n=1}^\infty\|\sqrt{n}S_n(f\otimes\Psi^{(n-1)})\|_{\otimes_\mathrm{s}^n\mathscr{H}}^2<\infty\right.\right\},\\
	(A^\dagger(f)\Psi)^{(0)}&:=0,\\
	(A^\dagger(f)\Psi)^{(n)}&:=\sqrt{n}S_n(f\otimes\Psi^{(n-1)}),\qquad n\in\mathbb{N},\ \Psi\in\operatorname{dom}(A^\dagger(f)),
\end{align*}
where $S_n$ is the symmetrization operator from $\otimes^n\mathscr{H}$ onto $\otimes_\mathrm{s}^n\mathscr{H}$ and $\operatorname{dom}(T)$ denotes the domain of an operator $T.$
Its adjoint $A(f):=A^\dagger(f)^*$ is called the \textit{annihilation operator}.
We have $A(f)^*=A^\dagger(f).$
Since $A(f)+A(f)^*$ is essentially self-adjoint,
we define a self-adjoint operator $\Phi_\mathrm{S}(f)$ by
\[
\Phi_\mathrm{S}(f) :=\frac{1}{\sqrt{2}}\overline{A(f)+A(f)^*},
\]
where $\overline{T}$ denotes the closure of an operator $T.$
We call $\Phi_\mathrm{S}(f)$ the \textit{Segal field operator} with test vector $f.$
\newline

\noindent
\textbf{\textit{Weyl representations.}}
Let $C$ be a conjugation on $\mathscr{H},$ that is, $C$ is an anti-linear, norm-preserving map on $\mathscr{H}$ satisfying $C^2 = 1$.
Define
\[
\mathscr{H}_C:=\{f\in\mathscr{H}\mid Cf=f\}.
\]
Then $\mathscr{H}_C$ becomes a real Hilbert space by restricting the inner product on $\mathscr{H}$ to $\mathscr{H}_C.$
Let $\mathscr{V}$ be a (not necessarily dense) real subspace of $\mathscr{H}_C.$
\begin{screen}
	\begin{center}
	Throughout the paper, we fix the triple $(\mathscr{H},C,\mathscr{V})$.
	\end{center}
\end{screen}
The symbol $\mathcal{S}_{C,\mathscr{V}}(\mathscr{H})$ denotes the set of all injective self-adjoint operators $T$ acting in $\mathscr{H}$ that satisfy the following two conditions:
\begin{itemize}
	\item $CT\subset TC,$
	\item $\mathscr{V}\subset \operatorname{dom}(T)\cap \operatorname{dom}(T^{-1}),$ and both $T\mathscr{V}$ and $T^{-1}\mathscr{V}$ are dense in $\mathscr{H}_C.$
\end{itemize}
For each $T\in\mathcal{S}_{C,\mathscr{V}}(\mathscr{H})$ and (not necessarily continuous) real-linear functionals $q,p:\mathscr{V}\to\mathbb{R},$
we set
\[
\phi_{T,q}(f):=\Phi_{\mathrm{S}}(T^{-1}f)+q(f),\qquad \pi_{T,p}(f):=\Phi_{\mathrm{S}}(\mathrm{i}Tf)+p(f),\qquad f\in\mathscr{V}
\]
and
\[
\rho_{T,q,p}:=\left\{\mathscr{F}_{\mathrm{b}}(\mathscr{H}),\{\phi_{T,q}(f),\pi_{T,p}(f)\mid f\in\mathscr{V}\}\right\}.
\]
Then, $\rho_{T,q,p}$ is an irreducible Weyl representation of the canonical commutation relations over $\mathscr{V}$.
This can be proved by an argument similar to that in \cite[Lemma 4.3]{MR3513945} or \cite[Theorem 5.47]{MR4812858}.

To establish a necessary and sufficient condition for the equivalence of two Weyl representations of the form $\rho_{T,q,p}$, we introduce a fundamental concept that transfers the ranges of $T_1^{\pm1}$ to those of $T_2^{\pm1}$. 

\begin{definition}
	Let $T_1,T_2\in\mathcal{S}_{C,\mathscr{V}}(\mathscr{H}).$
	A \textit{transfer pair} $(J_+,J_-)$ from $T_1$ to $T_2$ with respect to $(C,\mathscr{V})$ is a pair of bounded operators $J_+$ and $J_-$ on $\mathscr{H}$ such that
	\begin{equation*}
		J_+T_1f=T_2f,\qquad J_-T_1^{-1}f=T_2^{-1}f,\qquad\forall f\in\mathscr{V}.
	\end{equation*}
\end{definition}

\begin{theorem}\label{main thm}
	Let $T_1,T_2\in\mathcal{S}_{C,\mathscr{V}}(\mathscr{H}),$ and let $q_1,p_1,q_2,p_2:\mathscr{V}\to\mathbb{R}$ be real-linear functionals.
	Then, $\rho_{T_1,q_1,p_1}$ is equivalent to $\rho_{T_2,q_2,p_2}$ if and only if the following two conditions hold:
	\begin{itemize}
		\item there exists a transfer pair $(J_+,J_-)$ from $T_1$ to $T_2$ with respect to $(C,\mathscr{V}),$ and the difference $J_+-J_-$ is Hilbert--Schmidt,
		\item there exist $h_q,h_p\in\mathscr{H}_C$ such that
		\begin{equation}\label{q,p equivalence rel}
			q_2(f)=q_1(f)+\langle h_q,T_1^{-1}f\rangle,\qquad 	p_2(f)=p_1(f)+\langle h_p,T_1f\rangle,\qquad\forall f\in\mathscr{V}.
		\end{equation}
	\end{itemize}
In this case, we have
\begin{equation}\label{Uphi(f)U^*_and_Uphi(if)U^}
U\Phi_{\mathrm{S}}(f)U^*=\Phi_{\mathrm{S}}\left(J_-f\right)+\langle h_q,f\rangle,\qquad U\Phi_{\mathrm{S}}(\mathrm{i}f)U^*=\Phi_{\mathrm{S}}\left(\mathrm{i}J_+f\right)+\langle h_p,f\rangle,\qquad\forall f\in\mathscr{H}_C,
\end{equation}
where $U$ is a unitary operator on $\mathscr{F}_{\mathrm{b}}(\mathscr{H})$ satisfying
\begin{equation}\label{def of equiv of reprs rhoT1q1pq and rhoT2q2p2}
	U\phi_{T_1,q_1}(f)U^*=\phi_{T_2,q_2}(f),\qquad U\pi_{T_1,p_1}(f)U^*=\pi_{T_2,p_2}(f),\qquad \forall f\in\mathscr{V}.
\end{equation}
\end{theorem}

\begin{remark} 
	Theorem \ref{main thm} generalizes \cite[Theorem 5.1]{MR3513945}.
	Indeed, \cite[Theorem 5.1]{MR3513945} follows from Theorem \ref{main thm} by setting $J_+:=\overline{T_2T_1^{-1}}$ and $J_-:=\overline{T_2^{-1}T_1}$.
\end{remark}

For concrete examples of inequivalent Weyl representations of the form $\rho_{T,q,p}$, see, e.g., \cite{MR3513945}, \cite[Chapter 10]{MR4292535}, or \cite{MSS2026}.
\newline

\noindent
\textbf{\textit{Hamiltonians.}}
Let $S$ be a non-negative self-adjoint operator acting in $\mathscr{H}.$
The \textit{second quantization operator} $\mathrm{d}\Gamma_\mathrm{b}(S)$ of $S$ is a non-negative self-adjoint operator acting in $\mathscr{F}_{\mathrm{b}}(\mathscr{H}),$ which is defined by
\[
\mathrm{d}\Gamma_\mathrm{b}(S):=\bigoplus_{n=0}^\infty S^{(n)},
\]
where $S^{(0)}:=0$ and
\[
S^{(n)}:=\sum_{j=1}^n\underbrace{1\otimes\cdots\otimes1}_{j-1}\otimes \overset{j\text{th}}{S}\otimes\underbrace{1\otimes\cdots\otimes1}_{n-j},\qquad n\in\mathbb{N}.
\]

Let $g\in\operatorname{dom}(S^{-1/2})$ be given.
We define a Hamiltonian by
\[
H_S(g)
:=\mathrm{d}\Gamma_\mathrm{b}(S)+\Phi_{\mathrm{S}}(g).
\]
$H_S(g)$ is called a \textit{van Hove Hamiltonian} or a \textit{van Hove--Miyatake Hamiltonian}.
For a detailed study of the van Hove Hamiltonians, see, e.g., \cite[Section 10.9]{MR4292535}, \cite[Chapter 13]{MR4812858}, or \cite{MR2015428}.
In particular, $H_S(g)$ is self-adjoint on $\mathrm{dom}(\mathrm{d}\Gamma_\mathrm{b}(S))$ and bounded below.
\newline

\noindent
\textbf{\textit{Main results.}}
 Let $S$ be an injective non-negative self-adjoint operator acting in $\mathscr{H}$ with $CS\subset SC$, and $g\in\operatorname{dom}(S^{-1/2})$.
 Let $T\in\mathcal{S}_{C,\mathscr{V}}(\mathscr{H})$, and let $q,p:\mathscr{V}\to\mathbb{R}$ be real-linear functionals.
 Since $\rho_{T,q,p}$ is an irreducible Weyl representation, the triple
 \[
 \mathbb{M}(S,g,T,q,p):=\Big\{\mathscr{F}_{\mathrm{b}}(\mathscr{H}),H_{S}(g),\{\phi_{T,q}(f),\pi_{T,p}(f)\mid f\in\mathscr{V}\}\Big\}
 \]
 is an abstract Bose field model.
 Our main results are stated as follows.

\begin{theorem}\label{main thm2}
	Let $\mathbb{M}(S_1,g_1,T_1,q_1,p_1)$ and $\mathbb{M}(S_2,g_2,T_2,q_2,p_2)$ be two abstract Bose field models defined above.
	Then $\mathbb{M}(S_1,g_1,T_1,q_1,p_1)$ is equivalent to $\mathbb{M}(S_2,g_2,T_2,q_2,p_2)$ if and only if there exists a unitary operator $W$ on $\mathscr{H}$ such that the following three conditions hold:
	\begin{itemize}
		\item $WS_1W^*=S_2$,\qquad $CW=WC$,
		\item $WT_1f=T_2f,\qquad WT_1^{-1}f=T_2^{-1}f,\qquad\forall f\in\mathscr{V}$,
		\item $W^*g_2-g_1\in\operatorname{dom}{(S_1^{-1})}$ and 
		\begin{align*}
		q_2(f)-q_1(f)
		&=\operatorname{Re}\langle S_1^{-1}(W^*g_2-g_1),T_1^{-1}f\rangle,\\
		p_2(f)-p_1(f)
		&=-\operatorname{Im}\langle S_1^{-1}(W^*g_2-g_1),T_1f\rangle,\qquad\forall f\in\mathscr{V}.
		\end{align*}
	\end{itemize} 
\end{theorem}

\begin{corollary}\label{cor of main thm2}
	Let $\mathbb{M}(S_1,g_1,T_1,q_1,p_1)$ and $\mathbb{M}(S_2,g_2,T_2,q_2,p_2)$ be two abstract Bose field models defined above.
	We further suppose that one of the following conditions \textup{(A.1)} or \textup{(A.2)} holds:
	\begin{itemize}
		\item[\textup{(A.1)}] $T_1$ and $T_2$ are non-negative, and $\mathscr{V}+\mathrm{i}\mathscr{V}$ is a core for $T_1$ and $T_2$,
		\item[\textup{(A.2)}] $\mathscr{V}$ is dense in $\mathscr{H}_C$, and there exists a Borel function $\varphi:[0,\infty)\to\mathbb{R}$ such that $T_1=\varphi(S_1)$ and $T_2=\varphi(S_2)$.
	\end{itemize}
	If $\mathbb{M}(S_1,g_1,T_1,q_1,p_1)$ is equivalent to $\mathbb{M}(S_2,g_2,T_2,q_2,p_2)$, then we have $T_1=T_2$.
	In particular, the unitary operator $W$ appearing in Theorem \ref{main thm2} coincides with the identity operator on $\mathscr{H}$.
\end{corollary}

\noindent
\textbf{\textit{Organization of the paper.}}
This paper is organized as follows.
In Section \ref{sect;Properties of transfer pairs}, we investigate basic properties of transfer pairs.
In Section \ref{sect;proof of main}, we prove Theorem \ref{main thm}.
In Section \ref{sect;proof of main2}, we prove Theorem \ref{main thm2}.
In Section \ref{set;proof of cor}, we prove Corollary \ref{cor of main thm2}.
In Section \ref{set;examples}, we classify abstract free Bose field models, abstract van Hove--Miyatake models, infrared-renormalized van Hove--Miyatake models, and quadratic interaction models.
In the Appendix, we collect some results from operator theory and auxiliary results used in this paper.

\section{Properties of transfer pairs}\label{sect;Properties of transfer pairs}

In this section, we investigate basic properties of transfer pairs.

\begin{proposition} 
	Let $T_1,T_2\in\mathcal{S}_{C,\mathscr{V}}(\mathscr{H})$. 
	If a transfer pair from $T_1$ to $T_2$ with respect to $(C,\mathscr{V})$ exists, then it is unique. 
\end{proposition}

\begin{proof}
	Let $(J_+,J_-)$ and $(J_+^\prime,J_-^\prime)$ be transfer pairs from $T_1$ to $T_2$ with respect to $(C,\mathscr{V})$.
	By the definition of a transfer pair, we have
	\[
	J_+(T_1f)=T_2f=J_+^\prime (T_1f),\qquad\forall f\in\mathscr{V}.
	\]
	Since $T_1\mathscr{V}$ is dense in $\mathscr{H}_C$ and both $J_+$ and $J_+^\prime$ are bounded, we obtain $J_+f=J_+^\prime f$ for all $f\in\mathscr{H}_C$.
	Since every $f\in\mathscr{H}$ can be decomposed as $f=f_1+\mathrm{i}f_2,$ where
	\begin{equation}\label{def of f_1 and f_2}
		f_1:=\frac{1}{2}(f+Cf)\qquad\text{and}\qquad f_2:=\frac{1}{2\mathrm{i}}(f-Cf)
	\end{equation}
	are vectors in $\mathscr{H}_C$, we conclude that $J_+=J_+^\prime$.
	Similarly, the equality $J_-=J_-^\prime$ follows.
\end{proof}

The following lemma gives the adjoints of $J_+$ and $J_-$.

\begin{lemma}\label{adjoints of J_- and I{+-}}
	Let $T_1,T_2\in\mathcal{S}_{C,\mathscr{V}}(\mathscr{H})$.
	Suppose that a transfer pair $(J_+,J_-)$ from $T_1$ to $T_2$ with respect to $(C,\mathscr{V})$ exists.
	Then we have
	\[
	J_-^*T_2f=T_1f,\qquad J_+^*T_2^{-1}f=T_1^{-1}f,\qquad\forall f\in\mathscr{V}.
	\]
	That is, the pair $(J_-^*,J_+^*)$ is a transfer pair from $T_2$ to $T_1$ with respect to $(C,\mathscr{V}).$ 
	Moreover, $J_+$ and $J_-$ have bounded inverses, given by $J_+^{-1}=J_-^*$ and $J_-^{-1}=J_+^*.$
\end{lemma}

\begin{proof}
	For any $f,g\in\mathscr{V}$, we have
	\[
	\langle	T_1^{-1}g,J_-^*T_2f\rangle=\langle	J_-T_1^{-1}g,T_2f\rangle=\langle	T_2^{-1}g,T_2f\rangle=\langle	g,f\rangle=\langle	T_1^{-1}g,T_1f\rangle.
	\]
	Since $T_1^{-1}\mathscr{V}$ is dense in $\mathscr{H}_C,$ we obtain $J_-^*T_2f=T_1f.$
	Similarly, for any $f,g\in\mathscr{V}$, we have
		\[
	\langle	T_1g,J_+^*T_2^{-1}f\rangle=\langle	J_+T_1g,T_2^{-1}f\rangle=\langle	T_2g,T_2^{-1}f\rangle=\langle	g,f\rangle=\langle	T_1g,T_1^{-1}f\rangle.
	\]
	Since $T_1\mathscr{V}$ is dense in $\mathscr{H}_C,$ we deduce that $J_+^*T_2^{-1}f=T_1^{-1}f.$
	Thus $(J_-^*,J_+^*)$ is a transfer pair from $T_2$ to $T_1$ with respect to $(C,\mathscr{V}).$
	
	We next show that $J_+$ and $J_-$ have bounded inverses.
	For any $f\in\mathscr{V},$ we have
	\[
	J_-^*J_+T_1f=J_-^*T_2f=T_1f.
	\]
	Since $T_1\mathscr{V}$ is dense in $\mathscr{H}_C,$ it holds that $J_-^*J_+=1.$
	Similarly, $J_+J_-^*=1$ follows.
	Thus, $J_+$ has a bounded inverse with $J_+^{-1}=J_-^*.$
	By taking the adjoints of these two equalities, we deduce that $J_-$ also has a bounded inverse with $J_-^{-1}=J_+^*.$
	This completes the proof. 
\end{proof}

\begin{proposition} 
	Let $T_1,T_2\in\mathcal{S}_{C,\mathscr{V}}(\mathscr{H})$. 
	Suppose that a transfer pair $(J_+,J_-)$ from $T_1$ to $T_2$ with respect to $(C,\mathscr{V})$ exists.
	Then, $J_+\mathscr{H}_C=\mathscr{H}_C$ and $J_-\mathscr{H}_C=\mathscr{H}_C$.
\end{proposition}

\begin{proof}
	We first note that
	\[
	J_+(T_1f)=T_2f\in\mathscr{H}_C,\qquad\forall f\in\mathscr{V}.
	\]
	Since $T_1\mathscr{V}$ is dense in $\mathscr{H}_C$ and $\mathscr{H}_C$ is closed in $\mathscr{H}$, we have $J_+\mathscr{H}_C\subset\mathscr{H}_C$.
	Similarly, we obtain $J_-\mathscr{H}_C\subset\mathscr{H}_C$.
	
	By Lemma \ref{adjoints of J_- and I{+-}}, $(J_-^*,J_+^*)$ is a transfer pair from $T_2$ to $T_1$.
	Combining this with the previous paragraph, we obtain $J_-^*\mathscr{H}_C\subset\mathscr{H}_C$ and $J_+^*\mathscr{H}_C\subset\mathscr{H}_C$.
	By Lemma \ref{adjoints of J_- and I{+-}}, we have $J_+J_-^*=1$ and $J_-J_+^*=1$, and thus we obtain
	\[
	\mathscr{H}_C=J_+J_-^*\mathscr{H}_C\subset J_+\mathscr{H}_C\qquad\text{and}\qquad
	\mathscr{H}_C=J_-J_+^*\mathscr{H}_C\subset J_-\mathscr{H}_C.
	\]
	This, together with the previous paragraph, leads to $J_+\mathscr{H}_C=\mathscr{H}_C$ and $J_-\mathscr{H}_C=\mathscr{H}_C$.
\end{proof}

\begin{lemma}\label{C and J_pm commute}
	Let $T_1,T_2\in\mathcal{S}_{C,\mathscr{V}}(\mathscr{H}).$
	Suppose that a transfer pair $(J_+,J_-)$ from $T_1$ to $T_2$ with respect to $(C,\mathscr{V})$ exists.
	Then, we have
	\[
	CJ_+=J_+C,\qquad 	CJ_-=J_-C.
	\]
\end{lemma}

\begin{proof}
	For any $f_1,f_2\in\mathscr{V}$, it holds that
	\[
	CJ_+(T_1f_1+\mathrm{i}T_1f_2)=C(T_2f_1+\mathrm{i}T_2f_2)=T_2f_1-\mathrm{i}T_2f_2.
	\]
	On the other hand, we have
	\[
	J_+C(T_1f_1+\mathrm{i}T_1f_2)=J_+(T_1f_1-\mathrm{i}T_1f_2)=T_2f_1-\mathrm{i}T_2f_2.
	\]
	Since the set 
	\[
	\{T_1f_1+\mathrm{i}T_1f_2\mid f_1,f_2\in\mathscr{V}\}
	\]
	is dense in $\mathscr{H},$ we obtain $CJ_+=J_+C.$
	Similarly, the equality $CJ_-=J_-C$ follows.
\end{proof}

\begin{lemma}\label{improved lemma 5.8}
	Let $T_1,T_2\in\mathcal{S}_{C,\mathscr{V}}(\mathscr{H}).$
	Suppose that a transfer pair $(J_+,J_-)$ from $T_1$ to $T_2$ with respect to $(C,\mathscr{V})$ exists.
	Define operators $X_+$ and $X_-$ by
	\[
	X_+:=\frac{1}{2}\left(J_-+J_+\right),\qquad X_-:=\frac{1}{2}\left(J_--J_+\right).
	\]
	Then, it holds that
	\[
	X_+^*X_+-X_-^*X_-=1,\qquad X_-^*X_+=X_+^*X_-
	\]
	and
	\[
	X_+X_+^*-X_-X_-^*=1,\qquad X_-X_+^*=X_+X_-^*.
	\]
	In particular, $X_+$ has a bounded inverse.
\end{lemma}

\begin{proof}
	For any $f,g\in\mathscr{H}$, it holds that
	\begin{align*}
		4\langle f,X_+^*X_+g\rangle&=4\langle X_+f,X_+g\rangle\\
		&=\langle J_-f,J_-g\rangle+\langle J_+f,J_-g\rangle+\langle J_-f,J_+g\rangle+\langle J_+f,J_+g\rangle.
	\end{align*}
By Lemma \ref{adjoints of J_- and I{+-}}, we have $J_+^*J_-=1.$
	Thus we deduce that
	\[
	4\langle f,X_+^*X_+g\rangle
	=\langle J_-f,J_-g\rangle+2\langle f,g\rangle+\langle J_+f,J_+g\rangle.
	\]
	Similarly, we obtain
	\[
	4\langle f,X_-^*X_-g\rangle
	=\langle J_-f,J_-g\rangle-2\langle f,g\rangle+\langle J_+f,J_+g\rangle.
	\]
	Hence
	\[
	4\langle f,(X_+^*X_+-X_-^*X_-)g\rangle=4\langle f,g\rangle,
	\]
	which implies that $X_+^*X_+-X_-^*X_-=1$.
	
	We next show $X_-^*X_+=X_+^*X_-.$
	For any $f,g\in\mathscr{H}$, we have
	\begin{align*}
		4\langle f,X_-^*X_+g\rangle&=4\langle X_-f,X_+g\rangle\\
		&=\langle J_-f, J_-g\rangle-\langle  J_+f, J_-g\rangle+\langle  J_-f,J_+g\rangle-\langle J_+f,J_+g\rangle\\
		&=\langle J_-f,J_-g\rangle-\langle f,g\rangle+\langle f,g\rangle-\langle J_+f,J_+g\rangle\\
		&=\langle J_-f,J_-g\rangle-\langle J_+f,J_+g\rangle.
	\end{align*}
	Similarly, we obtain
	\[
	4\langle f,X_+^*X_-g\rangle=\langle J_-f,J_-g\rangle-\langle J_+f,J_+g\rangle
	=4\langle f,X_-^*X_+g\rangle,
	\]
	and thus we arrive at $X_-^*X_+=X_+^*X_-.$
	
	To prove the remaining two equalities, we first determine the adjoints of $X_+$ and $X_-.$
We have
	\[
	X_+^*=\frac{1}{2}\left(J_-^*+J_+^*\right)=\frac{1}{2}\left(J_+^*+J_-^*\right)
	\]
	and
	\[
	-X_{-}^*=-\frac{1}{2}\left(J_-^*-J_+^*\right)=\frac{1}{2}\left(J_+^*-J_-^*\right).
	\]
	By Lemma \ref{adjoints of J_- and I{+-}}, the operators $X_+^*$ and $-X_-^*$ correspond to $X_+$ and $X_-$ when $T_1$ and $T_2$ are interchanged.
	Thus, from the result of the first half, we obtain
	\[
	X_+X_+^*-X_-X_-^*=1,\qquad X_-X_+^*=X_+X_-^*.
	\]
	
	Finally, we demonstrate that $X_+$ has a bounded inverse.
	The equality $X_+^*X_+=1+X_-^*X_-$ shows that $X_+$ is injective.
	On the other hand, the equality $X_+X_+^*=1+X_-X_-^*$ shows that $X_+$ is surjective.
	This completes the proof.
\end{proof}

\section{Proof of Theorem \ref{main thm}}\label{sect;proof of main}
Throughout the proof, we write $\phi_{T,0}$, $\pi_{T,0}$ and $\rho_{T,0,0}$ simply as $\phi_{T}$, $\pi_{T}$ and $\rho_{T}$, respectively.
With this notation, we have
\[
\phi_{T,q}(f)=\phi_T(f)+q(f),\qquad\pi_{T,p}(f)=\pi_T(f)+p(f),\qquad\forall f\in\mathscr{V}.
\]

We first prove the ``if'' part.
We use the theory of Bogoliubov transformations.
Suppose that a transfer pair $(J_+,J_-)$ from $T_1$ to $T_2$ with respect to $(C,\mathscr{V})$ exists.
By Lemma \ref{C and J_pm commute} and Lemma \ref{improved lemma 5.8}, we obtain the following four equalities:
\[
X_{+}^*X_{+}-X_-^*X_-=1,\qquad X_-^*CX_+C=X_{+}^*CX_-C
\]
and
\[
X_+X_+^*-CX_-X_-^*C=1,\qquad CX_-X_+^*C=X_+X_-^*.
\]
For each $f\in\mathscr{H},$ we define an operator $B(f)$ acting in $\mathscr{F}_{\mathrm{b}}(\mathscr{H})$ by
\[
B(f):=\overline{A(X_+f)+A(CX_-f)^*}.
\]
The correspondence
\[
\{A(f),A(f)^*\mid f\in\mathscr{H}\} \mapsto \{B(f),B(f)^*\mid f\in\mathscr{H}\}
\]
is called a \textit{Bogoliubov transformation}.
It is known that there exists a unitary operator $U$ on $\mathscr{F}_{\mathrm{b}}(\mathscr{H})$ such that
\begin{equation}\label{bogo unitary def}
	UA(f)U^*=B(f),\qquad\forall f\in\mathscr{H}
\end{equation}
if and only if $X_-$ is Hilbert--Schmidt (see, e.g., \cite[Chapter 2]{MR2964269}, \cite[Theorem XI.108]{MR529429}, or \cite{MR516713}).

We now suppose that the difference $J_+-J_-,$ and consequently $X_-,$ is Hilbert--Schmidt.
Then there exists a unitary operator $U$ that satisfies \eqref{bogo unitary def}.
For any $f\in\mathscr{H}_C,$ it holds that
\begin{align*}
	\Phi_{\mathrm{S}}(f)
	&=\frac{1}{\sqrt{2}}\overline{[A(f)+A(f)^*]}
	\supset\frac{1}{\sqrt{2}}U^*\Big[B(f)+B(f)^*\Big]\restriction_{\mathscr{F}_\mathrm{b,0}(\mathscr{H})}U\\
	&=\frac{1}{\sqrt{2}}U^*\Big[A\left(J_-f\right)+A\left(J_-f\right)^*\Big]\restriction_{\mathscr{F}_\mathrm{b,0}(\mathscr{H})}U\\
	&=U^*\Phi_{\mathrm{S}}\left(J_-f\right)\restriction_{\mathscr{F}_\mathrm{b,0}(\mathscr{H})}U.
\end{align*}
By taking the closure, we deduce
\[
\Phi_{\mathrm{S}}(f)=U^*\Phi_{\mathrm{S}}\left(J_-f\right)U,\qquad \forall f\in\mathscr{H}_C.
\]
Replacing $f$ with $T_1^{-1}f$ for $f\in\mathscr{V},$ we obtain
\[
\phi_{T_1}(f)=\Phi_{\mathrm{S}}(T_1^{-1}f)
=U^*\Phi_{\mathrm{S}}(T_2^{-1}f)U
=U^*\phi_{T_2}(f)U,\qquad \forall f\in\mathscr{V}.
\]
Similarly, we have
\[
\Phi_{\mathrm{S}}(\mathrm{i}f)=U^*\Phi_{\mathrm{S}}\left(\mathrm{i}J_+f\right)U,\qquad \forall f\in\mathscr{H}_C.
\]
Replacing $f$ with $T_1f$ for $f\in\mathscr{V},$ we arrive at
\[
\pi_{T_1}(f)=\Phi_{\mathrm{S}}(\mathrm{i}T_1f)=U^*\Phi_{\mathrm{S}}(\mathrm{i}T_2f)U=U^*\pi_{T_2}(f)U,\qquad \forall f\in\mathscr{V}.
\]
Hence $\rho_{T_1}$ is equivalent to $\rho_{T_2}.$ 

We further suppose that there exist vectors $h_q$ and $h_p$ in $\mathscr{H}_C$ that satisfy \eqref{q,p equivalence rel}, and show that $\rho_{T_1,q_1,p_1}$ is equivalent to $\rho_{T_2,q_2,p_2}$.
Let $V:=\mathrm{e}^{-\mathrm{i}\Phi_\mathrm{S}(h_p-\mathrm{i}h_q)}.$
Note that $V$ is a unitary operator.
It follows from Lemma \ref{lem;ccr_for_field_ops} that
\[
V^*\phi_{T_1}(f)V=\phi_{T_1}(f)-\langle h_q,T_1^{-1}f\rangle=\phi_{T_1}(f)+q_1(f)-q_2(f)
\] 
and
\[
V^*\pi_{T_1}(f)V=\pi_{T_1}(f)-\langle h_p,T_1f\rangle=\pi_{T_1}(f)+p_1(f)-p_2(f)
\]
for all $f\in\mathscr{V},$ which imply that
\[
V^*U^*\phi_{T_2,q_2}(f)UV=\phi_{T_1,q_1}(f),\qquad V^*U^*\pi_{T_2,p_2}(f)UV=\pi_{T_1,p_1}(f),\qquad\forall f\in\mathscr{V}.
\]
Thus, $\rho_{T_1,q_1,p_1}$ is equivalent to $\rho_{T_2,q_2,p_2}.$

We next prove the ``only if'' part.
Suppose that there exists a unitary operator $U$ on $\mathscr{F}_{\mathrm{b}}(\mathscr{H})$ that satisfies \eqref{def of equiv of reprs rhoT1q1pq and rhoT2q2p2}.
Let $q:=q_2-q_1$ and $p:=p_2-p_1.$
We first observe that
\begin{equation}\label{Phi_S(T_1^{-1}(g))=}
	U\Phi_{\mathrm{S}}(T_1^{-1}g)U^*+q_1(g)=U\phi_{T_1,q_1}(g)U^*=\phi_{T_2,q_2}(g)=\Phi_{\mathrm{S}}(T_2^{-1}g)+q_2(g),\qquad\forall g\in \mathscr{V}.
\end{equation}
Since the real subspace $T_1^{-1}\mathscr{V}$ is dense in $\mathscr{H}_C,$ 
for any $f\in\mathscr{H}_C,$ we can choose a sequence $\{g_n\}_{n=1}^\infty$ in $\mathscr{V}$ such that
$\{T_1^{-1}g_n\}_{n=1}^\infty$ converges to $f.$
Then, it follows from Lemma \ref{lem;expectation_value}, equality \eqref{Phi_S(T_1^{-1}(g))=}, and Lemma \ref{lem;conti_of_field_op} that
\begin{align*}
	&\mathrm{e}^{-\|T_2^{-1}(g_m-g_n)\|^2/4}
	=|\mathrm{e}^{-\|T_2^{-1}(g_m-g_n)\|^2/4}|
	=|\langle\Omega_\mathscr{H},\mathrm{e}^{\mathrm{i}\Phi_\mathrm{S}(T_2^{-1}(g_m-g_n))}\Omega_\mathscr{H}\rangle|\\
	&\qquad=|\langle U^*\Omega_\mathscr{H},\mathrm{e}^{\mathrm{i}\Phi_\mathrm{S}(T_1^{-1}(g_m-g_n))}U^*\Omega_\mathscr{H}\rangle\mathrm{e}^{-\mathrm{i}q(g_m-g_n)}|
	=|\langle U^*\Omega_\mathscr{H},\mathrm{e}^{\mathrm{i}\Phi_\mathrm{S}(T_1^{-1}(g_m-g_n))}U^*\Omega_\mathscr{H}\rangle|\\
	&\qquad=|\langle \mathrm{e}^{-\mathrm{i}\Phi_\mathrm{S}(T_1^{-1}g_m-f)}U^*\Omega_\mathscr{H},\mathrm{e}^{\mathrm{i}\Phi_\mathrm{S}(f-T_1^{-1}g_n)}U^*\Omega_\mathscr{H}\rangle|
	\xrightarrow{m,n\to\infty}1.
\end{align*}
Thus $\{T_2^{-1}g_n\}_{n=1}^\infty$ is a Cauchy sequence, and hence it has a limit.
We show that the limit does not depend on the choice of the sequence $\{g_n\}_{n=1}^\infty.$
For this, take an arbitrary sequence $\{h_n\}_{n=1}^\infty$ in $\mathscr{V}$ such that $\{T_1^{-1}h_n\}_{n=1}^\infty$ converges to $f.$
By the same argument as above, we have
\begin{align*}
	\mathrm{e}^{-\|T_2^{-1}(g_n-h_n)\|^2/4}
	=|\langle \mathrm{e}^{-\mathrm{i}\Phi_\mathrm{S}(T_1^{-1}g_n-f)}U^*\Omega_\mathscr{H},\mathrm{e}^{\mathrm{i}\Phi_\mathrm{S}(f-T_1^{-1}h_n)}U^*\Omega_\mathscr{H}\rangle|
	\xrightarrow{n\to\infty}1,
\end{align*}
whence the limit of $\{T_2^{-1}h_n\}_{n=1}^\infty$ coincides with that of $\{T_2^{-1}g_n\}_{n=1}^\infty.$

We next show that $\{q(g_n)\}_{n=1}^\infty$ is a Cauchy sequence and its limit does not depend on the choice of $\{g_n\}_{n=1}^\infty.$
By an argument similar to the above, we see that
\begin{align*}
\mathrm{e}^{\mathrm{i}tq(g_m-g_n)}
&=\mathrm{e}^{t^2\|T_2^{-1}(g_m-g_n)\|^2/4}\langle U^*\Omega_\mathscr{H},\mathrm{e}^{\mathrm{i}t\Phi_\mathrm{S}(T_1^{-1}(g_m-g_n))}U^*\Omega_\mathscr{H}\rangle\\
&=\mathrm{e}^{t^2\|T_2^{-1}(g_m-g_n)\|^2/4}\langle \mathrm{e}^{-\mathrm{i}t\Phi_\mathrm{S}(T_1^{-1}g_m-f)}U^*\Omega_\mathscr{H},\mathrm{e}^{\mathrm{i}t\Phi_\mathrm{S}(f-T_1^{-1}g_n)}U^*\Omega_\mathscr{H}\rangle
\xrightarrow{m,n\to\infty}1,
\end{align*}
for all $t\in\mathbb{R}$.
This, together with Lemma \ref{lem;cauchy}, implies that $\{q(g_n)\}_{n=1}^\infty$ is a Cauchy sequence.
To show that the limit does not depend on the choice of $\{g_n\}_{n=1}^\infty,$ take an arbitrary sequence $\{h_n\}_{n=1}^\infty$ in $\mathscr{V}$ such that $\{T_1^{-1}h_n\}_{n=1}^\infty$ converges to $f$.
Then
\[
\mathrm{e}^{\mathrm{i}tq(g_n-h_n)}
=\mathrm{e}^{t^2\|T_2^{-1}(g_n-h_n)\|^2/4}\langle \mathrm{e}^{-\mathrm{i}t\Phi_\mathrm{S}(T_1^{-1}g_n-f)}U^*\Omega_\mathscr{H},\mathrm{e}^{\mathrm{i}t\Phi_\mathrm{S}(f-T_1^{-1}h_n)}U^*\Omega_\mathscr{H}\rangle
\xrightarrow{n\to\infty}1
\]
for all $t\in\mathbb{R},$ whence the limit of $q(g_n)$ coincides with that of $q(h_n).$

We define
\[
J_-f:=\lim_{n\to\infty}T_2^{-1}g_n\in\mathscr{H}_C,\qquad Q(f):=\lim_{n\to\infty}q(g_n)\in\mathbb{R},\qquad f\in\mathscr{H}_C.
\]
Then $J_-$ is a real-linear operator on $\mathscr{H}_C,$ and it satisfies
\begin{equation}\label{J_-T_1^{-1}g=T_2^{-1}g}
J_-T_1^{-1}g=T_2^{-1}g,\qquad\forall g\in\mathscr{V}.
\end{equation}
On the other hand, $Q$ is a real-linear functional on $\mathscr{H}_C,$ and it satisfies
\begin{equation}\label{Q(T_1^{-1}g)}
	Q(T_1^{-1}g)=q(g),\qquad\forall g\in\mathscr{V}.
\end{equation}
Moreover, for any $t\in\mathbb{R},$ it follows from \eqref{Phi_S(T_1^{-1}(g))=} that
\begin{equation*}\label{before subs1}
	U\mathrm{e}^{\mathrm{i}t\Phi_{\mathrm{S}}(f)}U^*
	=\mathrm{s}\textrm{-}\!\!\!\lim_{n\to\infty}U\mathrm{e}^{\mathrm{i}t\Phi_{\mathrm{S}}(T_1^{-1}g_n)}U^*
	=\mathrm{s}\textrm{-}\!\!\!\lim_{n\to\infty}\mathrm{e}^{\mathrm{i}t\Phi_{\mathrm{S}}(T_2^{-1}g_n)}\mathrm{e}^{\mathrm{i}tq(g_n)}
	=\mathrm{e}^{\mathrm{i}t\Phi_{\mathrm{S}}(J_-f)}\mathrm{e}^{\mathrm{i}tQ(f)}.
\end{equation*}
Thus we obtain
\begin{equation}\label{eq2}
	U\Phi_{\mathrm{S}}(f)U^*=\Phi_{\mathrm{S}}\left(J_-f\right)+Q(f),\qquad\forall f\in\mathscr{H}_C.
\end{equation}
This in particular implies that $J_-$ and $Q$ are continuous.
Indeed, for any sequence $\{f_n\}_{n=1}^\infty$ in $\mathscr{H}_C$ that converges to some $f\in\mathscr{H}_C,$ we have
\[
\mathrm{e}^{-\|J_-(f_n-f)\|^2/4}
=|\langle\Omega_\mathscr{H},\mathrm{e}^{\mathrm{i}\Phi_\mathrm{S}(J_-(f_n-f))}\Omega_\mathscr{H}\rangle|
=|\langle U^*\Omega_\mathscr{H},\mathrm{e}^{\mathrm{i}\Phi_\mathrm{S}(f_n-f)}U^*\Omega_\mathscr{H}\rangle|
\xrightarrow{n\to\infty}1,
\]
and thus $\{J_-f_n\}_{n=1}^\infty$ converges to $J_-f.$
This means that $J_-$ is continuous.
Moreover, we obtain
\[
\mathrm{e}^{\mathrm{i}tQ(f_n-f)}
=\mathrm{e}^{t^2\|J_-(f_n-f)\|^2/4}\langle U^*\Omega_\mathscr{H},\mathrm{e}^{\mathrm{i}t\Phi_\mathrm{S}(f_n-f)}U^*\Omega_\mathscr{H}\rangle
\xrightarrow{n\to\infty}1
\]
for all $t\in\mathbb{R}$.
This, together with Lemma \ref{lem;conv}, yields that $\{Q(f_n)\}_{n=1}^\infty$ converges to $Q(f)$, whence $Q$ is continuous.

By Riesz's representation theorem, we find a vector $h_q\in\mathscr{H}_C$ that satisfies 
\[
Q(f)=\langle h_q,f\rangle,\qquad\forall f\in\mathscr{H}_C.
\]
This, combined with \eqref{Q(T_1^{-1}g)}, yields that
\begin{equation}\label{q_2-q_1=}
	q_2(g)-q_1(g)=\langle h_q,T_1^{-1}g\rangle,\qquad\forall g\in\mathscr{V}.
\end{equation}
Moreover, the first part of \eqref{Uphi(f)U^*_and_Uphi(if)U^} follows from \eqref{eq2}.

On the other hand, since every $f\in\mathscr{H}$ can be decomposed as $f=f_1+\mathrm{i}f_2$, where $f_1$ and $f_2$ are vectors in $\mathscr{H}_C$ defined by \eqref{def of f_1 and f_2},
we define
\[
J_-f:=J_-f_1+\mathrm{i}J_-f_2.
\]
Summing up the above arguments, we have constructed a complex-linear bounded operator $J_-$ on $\mathscr{H}$ and a vector $h_q\in\mathscr{H}_C$ that satisfy \eqref{J_-T_1^{-1}g=T_2^{-1}g}, \eqref{eq2}, and \eqref{q_2-q_1=}.

An argument similar to the above, based on the equality
\begin{equation}\label{Phi_S(T_1(g))=}
	U\Phi_{\mathrm{S}}(\mathrm{i}T_1g)U^*+p_1(g)=U\pi_{T_1,p_1}(g)U^*=\pi_{T_2,p_2}(g)=\Phi_{\mathrm{S}}(\mathrm{i}T_2g)+p_2(g),\qquad\forall g\in \mathscr{V},
\end{equation}
shows that there exist a complex-linear bounded operator $J_+$ on $\mathscr{H}$ and a vector $h_p\in\mathscr{H}_C$ satisfying
\begin{equation}\label{J_+T_1g=T_2g}
	J_+T_1g=T_2g,\qquad p_2(g)-p_1(g)=\langle h_p,T_1g\rangle,\qquad\forall g\in\mathscr{V}
\end{equation}
and
\begin{equation}\label{eq2.5}
U\Phi_{\mathrm{S}}(\mathrm{i}f)U^*=\Phi_{\mathrm{S}}\left(\mathrm{i}J_+f\right)+\langle h_p,f\rangle,\qquad\forall f\in\mathscr{H}_C.
\end{equation}
In particular, $(J_+,J_-)$ is a transfer pair from $T_1$ to $T_2$ with respect to $(C,\mathscr{V})$.

In the following, we prove that $J_+-J_-$ is Hilbert--Schmidt.
The following argument closely follows the proof of \cite[Lemma 5.12]{MR3513945}.
Let $W:=\mathrm{e}^{\mathrm{i}\Phi_\mathrm{S}(J_-h_p-\mathrm{i}J_+h_q)}U$.
Then, it follows from \eqref{eq2}, \eqref{eq2.5}, and Lemma \ref{lem;ccr_for_field_ops} that
\[
W\Phi_{\mathrm{S}}(f)W^*=\Phi_{\mathrm{S}}(J_-f),\qquad W\Phi_{\mathrm{S}}(\mathrm{i}f)W^*=\Phi_{\mathrm{S}}\left(\mathrm{i}J_+f\right),\qquad\forall f\in\mathscr{H}_C.
\]
From this, for any $f\in\mathscr{H}_C$ and $\Psi\in\mathscr{F}_\mathrm{b,0}(\mathscr{H}),$ we have
\begin{align*}
WA(f)\Psi&=W\frac{1}{\sqrt{2}}\left[\Phi_{\mathrm{S}}(f)+\mathrm{i}\Phi_{\mathrm{S}}(\mathrm{i}f)\right]\Psi\\
&=\frac{1}{\sqrt{2}}\left[\Phi_{\mathrm{S}}\left(J_-f\right)+\mathrm{i}\Phi_{\mathrm{S}}\left(\mathrm{i}J_+f\right)\right]W\Psi.
\end{align*}
Thus, for all $\Phi\in\mathscr{F}_\mathrm{b,0}(\mathscr{H}),$ we obtain
\begin{equation}\label{inner product comp1}
\langle W^*\Phi,A(f)\Psi\rangle=\langle W^*A(X_+f)^*\Phi+W^*A(CX_-f)\Phi,\Psi\rangle,
\end{equation}
where $X_+$ and $X_-$ are defined in Lemma \ref{improved lemma 5.8}.
Every $f\in\mathscr{H}$ can be decomposed as $f=f_1+\mathrm{i}f_2$, where $f_1$ and $f_2$ are vectors in $\mathscr{H}_C$ defined by \eqref{def of f_1 and f_2}.
This, combined with equality \eqref{inner product comp1}, yields that
\begin{align}\label{inner product comp2}
	\langle W^*\Phi,A(f)\Psi\rangle
	&=\langle W^*\Phi,A(f_1)\Psi\rangle-\mathrm{i}\langle W^*\Phi,A(f_2)\Psi\rangle\notag\\
	&=\langle W^*A(X_+f_1)^*\Phi+W^*A(CX_-f_1)\Phi,\Psi\rangle\notag\\
	&\qquad\qquad\qquad-\mathrm{i}\langle W^*A(X_+f_2)^*\Phi+W^*A(CX_-f_2)\Phi,\Psi\rangle\notag\\
	&=\langle W^*A(X_+f)^*\Phi+W^*A(CX_-f)\Phi,\Psi\rangle.
\end{align}
Since $\mathscr{F}_\mathrm{b,0}(\mathscr{H})$ is a core for $A(f)$, equality \eqref{inner product comp2} holds for all $\Psi\in \operatorname{dom}(A(f)).$
This implies that $W^*\Phi\in \operatorname{dom}(A(f)^*)$ and
\[
A(f)^*W^*\Phi=W^*A(X_+f)^*\Phi+W^*A(CX_-f)\Phi.
\]
Recall that $\Omega_{\mathscr{H}}$ is the Fock vacuum, which satisfies $A(f)\Omega_{\mathscr{H}}=0.$
From this, we have
\[
0=\langle A(f)^*W^*\Phi,\Omega_{\mathscr{H}}\rangle
=\langle A(X_+f)^*\Phi,W\Omega_{\mathscr{H}}\rangle+\langle A(CX_-f)\Phi,W\Omega_{\mathscr{H}}\rangle,
\]
and thus, by setting $\Omega:=W\Omega_\mathscr{H},$ we obtain
\[
\langle A(X_+f)^*\Phi,\Omega\rangle=-\langle A(CX_-f)\Phi,\Omega\rangle.
\]
Since $X_+$ has a bounded inverse by Lemma \ref{improved lemma 5.8}, replacing $f$ with $X_+^{-1}f$, we obtain
\[
\langle A(f)^*\Phi,\Omega\rangle=-\langle A(CX_-X_+^{-1}f)\Phi,\Omega\rangle,\qquad\forall f\in\mathscr{H},\ \Phi\in\mathscr{F}_\mathrm{b,0}(\mathscr{H}).
\]
By \cite[Proposition 3.3]{MR3513945} or \cite[Proposition 8.4]{MR4292535}, we conclude that $X_-X_+^{-1}$ is Hilbert--Schmidt, and thus so is $X_-$.
Therefore, $J_+-J_-$ is Hilbert--Schmidt as well.
This completes the proof.

\section{Proof of Theorem \ref{main thm2}}\label{sect;proof of main2}

We first prove the ``only if'' part.
Suppose that $\mathbb{M}(S_1,g_1,T_1,q_1,p_1)$ is equivalent to $\mathbb{M}(S_2,g_2,T_2,q_2,p_2)$.
Then there exist a unitary operator $U$ on $\mathscr{F}_{\mathrm{b}}(\mathscr{H})$ and a real number $E$ such that 
\[
UH_{S_1}(g_1)U^*=H_{S_2}(g_2)+E
\]
and
\eqref{def of equiv of reprs rhoT1q1pq and rhoT2q2p2} holds.
It follows from Theorem~\ref{main thm} that there exist a transfer pair $(J_+,J_-)$ from $T_1$ to $T_2$ with respect to $(C,\mathscr{V})$ and $h_q,h_p\in\mathscr{H}_C$ such that $J_+-J_-$ is Hilbert--Schmidt and \eqref{q,p equivalence rel} holds.

We begin by proving that $J_+$ is unitary,  $J_+=J_-$, and $J_+S_1J_+^*=S_2$.
By Lemma~\ref{lem;trans_of_field_op_by_vH_hamiltonian}, we have
\[
\mathrm{e}^{\mathrm{i}tH_{S_j}(g_j)}\Phi_{\mathrm{S}}(\mathrm{i}T_jf)\mathrm{e}^{-\mathrm{i}tH_{S_j}(g_j)}
=\Phi_\mathrm{S}(\mathrm{i}\mathrm{e}^{\mathrm{i}tS_j}T_jf)+\operatorname{Im}\left\langle S_j^{-1}(\mathrm{e}^{\mathrm{i}tS_j}-1)T_jf,g_j\right\rangle
\]
for all $t\in\mathbb{R}$, $f\in\mathscr{V}$, and $j=1,2$.
Thus, for any $t\in\mathbb{R}$ and $f\in\mathscr{V},$ we have
\begin{align*}
	U\mathrm{e}^{\mathrm{i}tH_{S_1}(g_1)}\pi_{T_1,p_1}(f)\mathrm{e}^{-\mathrm{i}tH_{S_1}(g_1)}U^*
	&=\mathrm{e}^{\mathrm{i}tH_{S_2}(g_2)}\pi_{T_2,p_2}(f)\mathrm{e}^{-\mathrm{i}tH_{S_2}(g_2)}\\
	&=\Phi_\mathrm{S}(\mathrm{i}\mathrm{e}^{\mathrm{i}tS_2}T_2f)+\operatorname{Im}\left\langle S_2^{-1}(\mathrm{e}^{\mathrm{i}tS_2}-1)T_2f,g_2\right\rangle+p_2(f).
\end{align*}
On the other hand, by \eqref{Uphi(f)U^*_and_Uphi(if)U^}, we compute that
\begin{align*}
	&U\mathrm{e}^{\mathrm{i}tH_{S_1}(g_1)}\pi_{T_1,p_1}(f)\mathrm{e}^{-\mathrm{i}tH_{S_1}(g_1)}U^*\\
	&=U\Phi_\mathrm{S}(\mathrm{i}\mathrm{e}^{\mathrm{i}tS_1}T_1f)U^*+\operatorname{Im}\left\langle S_1^{-1}(\mathrm{e}^{\mathrm{i}tS_1}-1)T_1f,g_1\right\rangle+p_1(f)\\
	&=\Phi_\mathrm{S}(-J_-\sin(tS_1)T_1f+\mathrm{i}J_+\cos(tS_1)T_1f)+\operatorname{Im}\left\langle S_1^{-1}(\mathrm{e}^{\mathrm{i}tS_1}-1)T_1f,g_1\right\rangle+p_1(f)\\
	&\qquad\qquad\qquad\qquad\qquad\qquad\qquad\qquad\qquad+\langle h_q,-\sin(tS_1)T_1f\rangle+\langle h_p,\cos(tS_1)T_1f\rangle.
\end{align*}
Hence, we obtain
\begin{equation}\label{vH prf eq1}
	\mathrm{i}\mathrm{e}^{\mathrm{i}tS_2}T_2f=-J_-\sin(tS_1)T_1f+\mathrm{i}J_+\cos(tS_1)T_1f
\end{equation}
and
\begin{align}\label{vH prf eq2}
	&\operatorname{Im}\left\langle S_2^{-1}(\mathrm{e}^{\mathrm{i}tS_2}-1)T_2f,g_2\right\rangle\notag\\
	&=\operatorname{Im}\left\langle S_1^{-1}(\mathrm{e}^{\mathrm{i}tS_1}-1)T_1f,g_1\right\rangle+\langle h_q,-\sin(tS_1)T_1f\rangle+\langle h_p,[\cos(tS_1)-1]T_1f\rangle
\end{align}
for all $t\in\mathbb{R}$ and $f\in\mathscr{V}.$
Using Lemma \ref{C and J_pm commute}, which allows us to compare the imaginary parts of both sides of \eqref{vH prf eq1}, we have
\[
\cos(tS_2)T_2f=J_+\cos(tS_1)T_1f,\qquad\forall t\in\mathbb{R},\  f\in\mathscr{V}.
\]
Taking the inner product of both sides with $g\in\mathscr{H}$, we obtain
\[
\int_0^\infty\cos(u\lambda)\,\mathrm{d}\langle g,E_{S_2}(\lambda)T_2f\rangle
=\int_0^\infty\cos(u\lambda)\,\mathrm{d}\langle J_+^*g,E_{S_1}(\lambda)T_1f\rangle,\qquad\forall u\in\mathbb{R},
\]
where $S_j=\int_0^\infty\lambda\,\mathrm{d}E_{S_j}(\lambda)$ denotes the spectral resolution of $S_j$ for each $j=1,2$.
Integrating both sides from $0$ to $t\in\mathbb{R}$ with respect to $u$, we have
\[
\langle g,\operatorname{sinc}(tS_2)J_+T_1f\rangle=\langle g,J_+\operatorname{sinc}(tS_1)T_1f\rangle,\qquad\forall t\in\mathbb{R},\ f\in\mathscr{V},\ g\in\mathscr{H},
\]
whence
\[
J_+\operatorname{sinc}(tS_1)=\operatorname{sinc}(tS_2)J_+,\qquad\forall t\in\mathbb{R}.
\] 
By Lemma \ref{sinc eq}, we deduce that $J_+S_1\subset S_2J_+$.

Taking the real parts of \eqref{vH prf eq1}, we obtain
\[
\sin(tS_2)T_2f=J_-\sin(tS_1)T_1f,\qquad\forall t\in\mathbb{R},\ f\in\mathscr{V}.
\]
This, together with the fact that $T_2f=J_+T_1f$, implies that $J_-\sin(tS_1)=\sin(tS_2)J_+$.
By Lemma \ref{sin eq}, we have $J_-S_1\subset S_2J_+$.
As we have already shown that $J_+S_1\subset S_2J_+,$ we obtain $J_+=J_-$.
Since $J_-^{-1}=J_+^*$ by Lemma~\ref{adjoints of J_- and I{+-}}, we conclude that $J_+$ is unitary.
Moreover, it follows from $J_+S_1\subset S_2J_+$ that $J_+S_1J_+^*=S_2$.

We next show that $J_+^*g_2-g_1$ belongs to $\operatorname{dom}(S_1^{-1})$.
Let
\begin{equation}\label{eq;Re and Im parts for g}
\operatorname{Re}(g_j):=\frac{1}{2}(g_j+Cg_j),\qquad \operatorname{Im}(g_j):=\frac{1}{2\mathrm{i}}(g_j-Cg_j),\qquad j=1,2.
\end{equation}
Note that, by Lemma~\ref{intertwining Borel}, we have
\[
S_2^{-1}(\mathrm{e}^{\mathrm{i}tS_2}-1)T_2f=J_+S_1^{-1}(\mathrm{e}^{\mathrm{i}tS_1}-1)T_1f,\qquad\forall t\in\mathbb{R},\ f\in\mathscr{V}.
\]
By \eqref{vH prf eq2} and the density of $T_1\mathscr{V}$ in $\mathscr{H}_C$, we obtain
\begin{align}\label{comp of imaginary part}
	&S_1^{-1}[\cos(tS_1)-1]J_+^*\operatorname{Im}(g_2)-S_1^{-1}\sin(tS_1)J_+^*\operatorname{Re}(g_2)\notag\\
	&\qquad=S_1^{-1}[\cos(tS_1)-1]\operatorname{Im}(g_1)-S_1^{-1}\sin(tS_1)\operatorname{Re}(g_1)-\sin(tS_1)h_q+[\cos(tS_1)-1]h_p.
\end{align}
Dividing both sides by $t\not=0$, taking the inner product of both sides with $\psi\in\mathrm{dom}(S_1)$, and letting $t\to0$, we deduce that
\begin{equation*}
\langle\psi,J_+^*\operatorname{Re}(g_2)\rangle=\langle\psi,\operatorname{Re}(g_1)\rangle+\langle S_1\psi,h_q\rangle,\qquad\forall \psi\in\mathrm{dom}(S_1).
\end{equation*}
This implies that $h_q\in\operatorname{dom}(S_1)$ and
\begin{equation}\label{J_+^*Reg_2}
	J_+^*\operatorname{Re}(g_2)-\operatorname{Re}(g_1)=S_1h_q.
\end{equation}
Thus, $J_+^*\operatorname{Re}(g_2)-\operatorname{Re}(g_1)\in\operatorname{dom}(S_1^{-1})$ and
\[
q_2(f)-q_1(f)
=\langle h_q,T_1^{-1}f\rangle
=\langle S_1^{-1}[J_+^*\operatorname{Re}(g_2)-\operatorname{Re}(g_1)],T_1^{-1}f\rangle,\qquad\forall f\in\mathscr{V}.
\]
On the other hand, substituting \eqref{J_+^*Reg_2} into \eqref{comp of imaginary part}, we obtain
\[
S_1^{-1}[\cos(tS_1)-1]J_+^*\operatorname{Im}(g_2)
=S_1^{-1}[\cos(tS_1)-1]\operatorname{Im}(g_1)+[\cos(tS_1)-1]h_p.
\]
Multiplying both sides from the left by $t^{-2}S_1^{-1}$ with $t\not=0$ and letting $t\to0$, we deduce that $h_p\in\operatorname{dom}(S_1)$ and
\[
J_+^*\operatorname{Im}(g_2)=\operatorname{Im}(g_1)+S_1h_p.
\]
Thus, $J_+^*\operatorname{Im}(g_2)-\operatorname{Im}(g_1)\in\operatorname{dom}(S_1^{-1})$ and
\[
p_2(f)-p_1(f)
=\langle h_p,T_1f\rangle
=\langle S_1^{-1}[J_+^*\operatorname{Im}(g_2)-\operatorname{Im}(g_1)],T_1f\rangle,\qquad\forall f\in\mathscr{V}.
\]
Summing up the above arguments, we conclude that $J_+^*g_2-g_1\in\operatorname{dom}(S_1^{-1})$ and
\[
q_2(f)-q_1(f)
=\operatorname{Re}\langle S_1^{-1}(J_+^*g_2-g_1),T_1^{-1}f\rangle,\quad
p_2(f)-p_1(f)
=-\operatorname{Im}\langle S_1^{-1}(J_+^*g_2-g_1),T_1f\rangle
\]
for all $f\in\mathscr{V}$.
Therefore, letting $W:=J_+$, we obtain the desired result.

We next prove the ``if'' part.
Let $V:=\mathrm{e}^{\mathrm{i}\Phi_\mathrm{S}(\mathrm{i}S_1^{-1}(W^*g_2-g_1))}$.
It follows from Lemma \ref{lem;ccr_for_field_ops} that
\[
V\phi_{T_1,q_1}(f)V^*=\Phi_{\mathrm{S}}(T_1^{-1}f)+\operatorname{Re}\langle S_1^{-1}(W^*g_2-g_1),T_1^{-1}f\rangle+q_1(f)=\Phi_{\mathrm{S}}(T_1^{-1}f)+q_2(f)
\] 
and
\[
V\pi_{T_1,p_1}(f)V^*=\Phi_{\mathrm{S}}(\mathrm{i}T_1f)-\operatorname{Im}\langle S_1^{-1}(W^*g_2-g_1),T_1f\rangle+p_1(f)=\Phi_{\mathrm{S}}(\mathrm{i}T_1f)+p_2(f)
\]
for all $f\in\mathscr{V}$.
Moreover, by Lemma \ref{lem;trans_of_second_quant_op_by_field_op} and Lemma \ref{lem;ccr_for_field_ops}, we have
\[
VH_{S_1}(g_1)V^*=H_{S_1}(W^*g_2)+E,
\]
where $E$ is some real number.
Let $\Gamma_\mathrm{b}(W)$ be a unitary operator defined by \eqref{eq;def_of_Gamma_b(W)}.
Then, by Lemma \ref{lem;trans_of_field_op_by_Gamma(U)}, the unitary operator $U:=\Gamma_\mathrm{b}(W)V$ satisfies \eqref{def of equiv of reprs rhoT1q1pq and rhoT2q2p2} and
\[
UH_{S_1}(g_1)U^*=H_{S_2}(g_2)+E.
\]
Hence $\mathbb{M}(S_1,g_1,T_1,q_1,p_1)$ is equivalent to $\mathbb{M}(S_2,g_2,T_2,q_2,p_2)$.
This completes the proof.

\section{Proof of Corollary \ref{cor of main thm2}}\label{set;proof of cor}

We first consider the case where \textup{(A.1)} is satisfied.
By Theorem \ref{main thm2}, we have
\[
\|T_1f\|=\|WT_1f\|=\|T_2f\|,\qquad \forall f\in\mathscr{V}.
\]
This leads to
\[
\|T_1f\|=\|T_2f\|,\qquad \forall f\in\mathscr{V}+\mathrm{i}\mathscr{V}.
\]
Since $\mathscr{V}+\mathrm{i}\mathscr{V}$ is a core for $T_1$ and $T_2$, a limiting argument yields that $\operatorname{dom}(T_1)=\operatorname{dom}(T_2)$ and
\[
\|T_1f\|=\|T_2f\|,\qquad \forall f\in\operatorname{dom}(T_1).
\] 
This, together with the polarization identity, implies that
\[
\langle f,T_1^2g\rangle=\langle T_2f,T_2g\rangle,\qquad\forall f\in\operatorname{dom}(T_2),\ g\in\operatorname{dom}(T_1^2).
\]
Thus, $T_1^2\subset T_2^2$.
By a similar argument, we have $T_2^2\subset T_1^2$, whence $T_1^2=T_2^2$.
Since $T_1$ and $T_2$ are non-negative, we obtain $T_1=T_2$.

We next consider the case where \textup{(A.2)} is satisfied.
By Theorem \ref{main thm2}, we have
\begin{equation*}
W\sin(tS_1)T_1^{-1}g=\sin(tS_2)T_2^{-1}g,\qquad \forall t\in\mathbb{R},\ g\in\mathscr{V}.
\end{equation*}
From this and the relation $T_2=\varphi(S_2)$, we obtain
\[
\langle T_2f,W\sin(tS_1)T_1^{-1}g\rangle
=\langle T_2f,\sin(tS_2)T_2^{-1}g\rangle
=\langle f,\sin(tS_2)g\rangle,\qquad\forall f,g\in\mathscr{V}.
\]
On the other hand, it follows from Theorem \ref{main thm2} and the relation $T_1=\varphi(S_1)$ that
\begin{align*}
\langle T_2f,W\sin(tS_1)T_1^{-1}g\rangle
&=\langle W^*T_2f,\sin(tS_1)T_1^{-1}g\rangle\\
&=\langle T_1f,\sin(tS_1)T_1^{-1}g\rangle
=\langle f,\sin(tS_1)g\rangle,\qquad\forall f,g\in\mathscr{V}.
\end{align*}
Since $\mathscr{V}$ is dense in $\mathscr{H}_C$, we obtain $\sin(tS_1)=\sin(tS_2)$.
By Lemma \ref{sin eq}, this implies that $S_1\subset S_2$, whence $S_1=S_2$.
Thus, $T_1=T_2$.

\section{Examples}\label{set;examples}

In this section, we apply our main results, Theorem \ref{main thm2} and Corollary \ref{cor of main thm2}, to concrete models and classify them up to equivalence.

\subsection{The abstract free Bose field model}
Let $S$ be an injective non-negative self-adjoint operator acting in $\mathscr{H}$ with $S^{1/2}\in\mathcal{S}_{C,\mathscr{V}}(\mathscr{H})$.
The model
\[
\mathbb{M}_{S}^{\mathrm{free}}:=\mathbb{M}(S,0,S^{1/2},0,0)
\]
is called an \textit{abstract free Bose field model} (see \cite[Example 10.2]{MR4292535} or \cite[Section 5.18 and p.~844]{MR4812858}).
Partial results on the classification of abstract free Bose field models are given in \cite[Example 10.2]{MR4292535}.
Here, we provide a complete classification under the assumption that $\mathscr{V}$ is dense in $\mathscr{H}_C$.

\begin{theorem}\label{thm;free Bose field} 
	Suppose that $\mathscr{V}$ is dense in $\mathscr{H}_C$. Let $\mathbb{M}_{S_1}^{\mathrm{free}}$ and $\mathbb{M}_{S_2}^{\mathrm{free}}$ be two abstract free Bose field models. 
	Then $\mathbb{M}_{S_1}^{\mathrm{free}}$ is equivalent to $\mathbb{M}_{S_2}^{\mathrm{free}}$ if and only if $S_1=S_2$. 
\end{theorem}

\begin{proof}
	Since condition (A.2) is satisfied, the theorem follows from Corollary \ref{cor of main thm2}.
\end{proof}

\subsection{The abstract van Hove--Miyatake model}

Let $S$ be an injective non-negative self-adjoint operator acting in $\mathscr{H}$ with $S^{1/2}\in\mathcal{S}_{C,\mathscr{V}}(\mathscr{H})$, and let $g\in\operatorname{dom}(S^{-1/2})$.
The model
\[
\mathbb{M}_{S,g}^{\mathrm{vHM}}:=\mathbb{M}(S,g,S^{1/2},0,0)
\]
is called an \textit{abstract van Hove--Miyatake model} (see \cite[Definition 10.4]{MR4292535} or \cite[Section 13.4]{MR4812858}).

\begin{theorem}
	Suppose that $\mathscr{V}$ is dense in $\mathscr{H}_C$.
	Let $\mathbb{M}_{S_1,g_1}^{\mathrm{vHM}}$ and $\mathbb{M}_{S_2,g_2}^{\mathrm{vHM}}$ be two abstract van Hove--Miyatake models.
	Then $\mathbb{M}_{S_1,g_1}^{\mathrm{vHM}}$ is equivalent to $\mathbb{M}_{S_2,g_2}^{\mathrm{vHM}}$ if and only if $S_1=S_2$ and $g_1=g_2$.
\end{theorem}

\begin{proof}
	It suffices to show the ``only if'' part.
	Suppose that $\mathbb{M}_{S_1,g_1}^{\mathrm{vHM}}$ is equivalent to $\mathbb{M}_{S_2,g_2}^{\mathrm{vHM}}$.
	Since condition (A.2) of Corollary \ref{cor of main thm2} is satisfied, we can apply Theorem \ref{main thm2} and Corollary \ref{cor of main thm2} to obtain $S_1=S_2$, $g_2-g_1\in\operatorname{dom}(S_1^{-1})$, and
	\[
	\operatorname{Re}\langle S_1^{-1}(g_2-g_1),S_1^{-1/2}f\rangle=0,\qquad \operatorname{Im}\langle S_1^{-1}(g_2-g_1),S_1^{1/2}f\rangle=0,\qquad\forall f\in\mathscr{V}.
	\]
	Since $S_1^{1/2}\mathscr{V}$ and $S_1^{-1/2}\mathscr{V}$ are dense in $\mathscr{H}_C$, we have
	\[
	\langle S_1^{-1}(g_2-g_1),f\rangle=0,\qquad\forall f\in\mathscr{H}.
	\]
	This implies that $g_1=g_2$.
\end{proof}

\subsection{The infrared-renormalized van Hove--Miyatake model}

Let $S$ be an injective non-negative self-adjoint operator acting in $\mathscr{H}$ such that $S^{1/2}\in\mathcal{S}_{C,\mathscr{V}}(\mathscr{H})$,  $\mathscr{V}\subset\operatorname{dom}(S^{-1})$, and $S^{-1}\mathscr{V}$ is dense in $\mathscr{H}_C$.
Let $g\in\operatorname{dom}(S^{-1/2})$.
The model
\[
\mathbb{M}_{S,g}^{\mathrm{IR}}:=\mathbb{M}(S,0,S^{1/2},q_{S,g},p_{S,g}),
\]
where
\[
q_{S,g}(f):=-\operatorname{Re}\langle S^{-1/2}g,S^{-1}f\rangle,\qquad p_{S,g}(f):=\operatorname{Im}\langle S^{-1/2}g,f\rangle,\qquad f\in\mathscr{V},
\]
is called an \textit{infrared-renormalized van Hove--Miyatake model} (see \cite[Section 10.9.4]{MR4292535} or \cite[Example 13.12]{MR4812858}).

\begin{theorem}\label{thm IRvHM}
	Suppose that $\mathscr{V}$ is dense in $\mathscr{H}_C$.
	Let $\mathbb{M}_{S_1,g_1}^{\mathrm{IR}}$ and $\mathbb{M}_{S_2,g_2}^{\mathrm{IR}}$ be two infrared-renormalized van Hove--Miyatake models.
	Then $\mathbb{M}_{S_1,g_1}^{\mathrm{IR}}$ is equivalent to $\mathbb{M}_{S_2,g_2}^{\mathrm{IR}}$ if and only if $S_1=S_2$ and $g_1=g_2$.
\end{theorem}

\begin{proof}
	It suffices to show the ``only if'' part.
	Suppose that $\mathbb{M}_{S_1,g_1}^{\mathrm{IR}}$ is equivalent to $\mathbb{M}_{S_2,g_2}^{\mathrm{IR}}$.
	Since condition (A.2) of Corollary \ref{cor of main thm2} is satisfied, we can apply Theorem \ref{main thm2} and Corollary \ref{cor of main thm2} to obtain $S_1=S_2$ and
	\[
	\operatorname{Re}\langle S_1^{-1/2}(g_2-g_1),S_1^{-1}f\rangle=0,\qquad \operatorname{Im}\langle S_1^{-1/2}(g_2-g_1),f\rangle=0,\qquad\forall f\in\mathscr{V}.
	\]
	Since $S_1^{-1}\mathscr{V}$ and $\mathscr{V}$ are dense in $\mathscr{H}_C$, we have
	\[
	\langle S_1^{-1/2}(g_2-g_1),f\rangle=0,\qquad\forall f\in\mathscr{H}.
	\]
	This implies that $g_1=g_2$.
\end{proof}

The following corollary generalizes \cite[Corollary 10.7]{MR4292535} and answers a question of A.~Arai \cite[Remark 10.19]{MR4292535}.

\begin{corollary}
	Let $S$ be an injective non-negative self-adjoint operator acting in $\mathscr{H}$ with $CS\subset SC$.
	Suppose that
	\[
	\mathscr{V}=\operatorname{dom}(S^{-1})\cap\operatorname{dom}(S^{1/2})\cap\mathscr{H}_C.
	\]
	Let $g_1,g_2\in\operatorname{dom}(S^{-1/2})$ be arbitrary.
	Then $\mathbb{M}_{S,g_1}^{\mathrm{IR}}$ is equivalent to $\mathbb{M}_{S,g_2}^{\mathrm{IR}}$ if and only if $g_1=g_2$.
\end{corollary}

\begin{proof}
	This follows from Theorem \ref{thm IRvHM}.
\end{proof}

\subsection{The quadratic interaction model}

Let $(S,g=\{g_n\}_{n=1}^\infty,\lambda=\{\lambda_n\}_{n=1}^\infty)$ be a triple satisfying the following five conditions:
\begin{itemize}
	\item  $S$ is an injective non-negative self-adjoint operator acting in $\mathscr{H}$ with $S^{1/2}\in\mathcal{S}_{C,\mathscr{V}}(\mathscr{H})$,
	\item $\lambda_n\in\mathbb{R}$ and $g_n\in\operatorname{dom}(S^{1/2})\cap\operatorname{dom}(S^{-1/2})\cap\mathscr{H}_C$ for all $n\in\mathbb{N}$,
	\item $\sum_{n=1}^\infty|\lambda_n|\cdot\|S^{-1/2}g_n\|^2<\infty$,
	\item $\sum_{n=1}^\infty|\lambda_n|\cdot\|S^{1/2}g_n\|^2<\infty$,
	\item for some $\varepsilon>0$, the operator inequality
	\[
	1+\sum_{n=1}^\infty\lambda_n|S^{-1/2}g_n\rangle\langle S^{-1/2}g_n|\geq\varepsilon
	\]
	holds.
\end{itemize}
Here, for each $\psi,\phi\in\mathscr{H}$, a rank-one operator $|\psi\rangle\langle\phi|$ is defined by
\[
|\psi\rangle\langle\phi|f:=\langle\phi,f\rangle\psi,\qquad f\in\mathscr{H}.
\]

\begin{remark}
	The above five conditions correspond to conditions (B1)--(B5) in \cite[Section~4]{MR4213757}.
	Condition (B6) is automatically satisfied by setting $J:=C$.
	Indeed, since $S^{1/2}\in\mathcal{S}_{C,\mathscr{V}}(\mathscr{H})$, we have $CSC=S$.
	Moreover, since $g_n\in\mathscr{H}_C$, we have $Cg_n=g_n$ for all $n\in\mathbb{N}$.
\end{remark}

We set
\[
H^{\mathrm{QI}}_{S,g,\lambda}:=\mathrm{d}\Gamma_\mathrm{b}(S)+\frac{1}{2}\sum_{n=1}^\infty\lambda_n\Phi_{\mathrm{S}}(g_n)^2,
\]
which acts in $\mathscr{F}_{\mathrm{b}}(\mathscr{H})$.
The operator $H^{\mathrm{QI}}_{S,g,\lambda}$ has been studied in \cite{MR4298882,gamet2026renormalization,IIMS2026,MR4213757}.
By \cite[Theorem 4.3]{MR4213757}, $H^{\mathrm{QI}}_{S,g,\lambda}$ is self-adjoint.
Thus, the triple
\[
\mathbb{M}^{\mathrm{QI}}_{S,g,\lambda}:=\Big\{\mathscr{F}_{\mathrm{b}}(\mathscr{H}),H^{\mathrm{QI}}_{S,g,\lambda},\{\phi_{S^{1/2},0}(f),\pi_{S^{1/2},0}(f)\mid f\in\mathscr{V}\}\Big\}
\]
defines an abstract Bose field model.
We call $\mathbb{M}^{\mathrm{QI}}_{S,g,\lambda}$ a \textit{quadratic interaction model}.

To classify quadratic interaction models, we introduce
\[
\gamma(S,g,\lambda):=S^2+\sum_{n=1}^\infty\lambda_n|S^{1/2}g_n\rangle\langle S^{1/2}g_n|.
\]

\begin{theorem}\label{thm;quadratic interaction}
	Suppose that $\mathscr{V}$ is dense in $\mathscr{H}_C$.
	Let $\mathbb{M}_{S_1,g_1,\lambda_1}^{\mathrm{QI}}$ and $\mathbb{M}_{S_2,g_2,\lambda_2}^{\mathrm{QI}}$ be two quadratic interaction models. 
	Then $\mathbb{M}_{S_1,g_1,\lambda_1}^{\mathrm{QI}}$ is equivalent to $\mathbb{M}_{S_2,g_2,\lambda_2}^{\mathrm{QI}}$ if and only if 
	\[
      \gamma(S_1,g_1,\lambda_1)=\gamma(S_2,g_2,\lambda_2).
	\]
	That is, the map $\gamma$ is a complete invariant.
\end{theorem}

\begin{lemma}\label{lem;quadratic interaction}
	Any quadratic interaction model $\mathbb{M}^{\mathrm{QI}}_{S,g,\lambda}$ is equivalent to an abstract free Bose field model $\mathbb{M}_{K}^{\mathrm{free}}$, where
	\[
	K:=\gamma(S,g,\lambda)^{1/2}
	=\sqrt{S^2+\sum_{n=1}^\infty\lambda_n|S^{1/2}g_n\rangle\langle S^{1/2}g_n|}.
	\]
\end{lemma}

\begin{proof}
	We first note that, by \cite[Lemma 5.1]{MR4213757}, the operator $\gamma(S,g,\lambda)$ is non-negative and injective.
	Thus, $K$ is well-defined and injective.
	
	We next show that $K^{1/2}\in\mathcal{S}_{C,\mathscr{V}}(\mathscr{H})$.
	Since $CS\subset SC$ and $Cg_n=g_n$, we have $CK\subset KC$.
	Moreover, by \cite[Lemma 5.2]{MR4213757} and \cite[Lemma 3.1]{MR4213757}, we obtain $\operatorname{dom}(S^{1/2})=\operatorname{dom}(K^{1/2})$ and $\operatorname{dom}(S^{-1/2})=\operatorname{dom}(K^{-1/2})$, whence $\mathscr{V}\subset\operatorname{dom}(K^{1/2})\cap\operatorname{dom}(K^{-1/2})$.
	Furthermore, by \cite[Lemma 5.2]{MR4213757}, \cite[Lemma 3.1]{MR4213757}, and \cite[Lemma 3.2]{MR4213757}, $\overline{K^{1/2}S^{-1/2}}$ and $\overline{K^{-1/2}S^{1/2}}$ are bounded and bijective with
	\[
	\left(\overline{K^{1/2}S^{-1/2}}\right)^{-1}=\overline{S^{1/2}K^{-1/2}},\qquad\left(\overline{K^{-1/2}S^{1/2}}\right)^{-1}=\overline{S^{-1/2}K^{1/2}},
	\]
	and thus,
	\[
	K^{1/2}\mathscr{V}=	\overline{K^{1/2}S^{-1/2}}\cdot S^{1/2}\mathscr{V}\qquad\text{and}\qquad
	K^{-1/2}\mathscr{V}=\overline{K^{-1/2}S^{1/2}}\cdot S^{-1/2}\mathscr{V}
	\]
	are dense in $\mathscr{H}_C$.
	
	Finally, we show that $\mathbb{M}^{\mathrm{QI}}_{S,g,\lambda}$ is equivalent to $\mathbb{M}_{K}^{\mathrm{free}}$.
	By \cite[Theorem 5.3]{MR4213757}, there exist a unitary operator $U$ on $\mathscr{F}_\mathrm{b}(\mathscr{H})$ and a real number $E$ such that
	\begin{equation}\label{eq;quadratic interaction_eq0}
	UH^{\mathrm{QI}}_{S,g,\lambda}U^*=\mathrm{d}\Gamma_\mathrm{b}(K)+E
	\end{equation}
	and
	\begin{equation}\label{eq;quadratic interaction_eq1}
	U\left[\overline{A(Xf)+A(CYf)^*}\right]U^*=A(f),\qquad\forall f\in\mathscr{H},
	\end{equation}
	where $X$ and $Y$ are defined by
	\[
	X:=\frac{1}{2}\left(\overline{S^{-1/2}K^{1/2}}+\overline{S^{1/2}K^{-1/2}}\right),\qquad Y:=\frac{1}{2}\left(\overline{S^{-1/2}K^{1/2}}-\overline{S^{1/2}K^{-1/2}}\right).
	\]
 	It follows from \eqref{eq;quadratic interaction_eq1} and \cite[Lemma 2.3]{MR4213757} that
	\begin{equation}\label{eq;quadratic interaction_eq2}
	U\Phi_\mathrm{S}(Xf+CYf)U^*=\Phi_\mathrm{S}(f),\qquad\forall f\in\mathscr{H}.
	\end{equation}
	Replacing $f$ in \eqref{eq;quadratic interaction_eq2} with $K^{-1/2}f$ for $f\in\mathscr{V}$ yields
	\begin{equation}\label{eq;quadratic interaction_eq3}
		U\Phi_\mathrm{S}(S^{-1/2}f)U^*=\Phi_\mathrm{S}(K^{-1/2}f),\qquad\forall f\in\mathscr{V}.
	\end{equation}
	Moreover, replacing $f$ in \eqref{eq;quadratic interaction_eq2} with $\mathrm{i}K^{1/2}f$ for $f\in\mathscr{V}$ yields
	\begin{equation}\label{eq;quadratic interaction_eq4}
		U\Phi_\mathrm{S}(\mathrm{i}S^{1/2}f)U^*=\Phi_\mathrm{S}(\mathrm{i}K^{1/2}f),\qquad\forall f\in\mathscr{V}.
	\end{equation}
	Equalities \eqref{eq;quadratic interaction_eq0}, \eqref{eq;quadratic interaction_eq3}, and \eqref{eq;quadratic interaction_eq4} imply that $\mathbb{M}^{\mathrm{QI}}_{S,g,\lambda}$ is equivalent to $\mathbb{M}_{K}^{\mathrm{free}}$.
\end{proof}

\begin{proof}[Proof of Theorem \ref{thm;quadratic interaction}]
	This follows from Lemma \ref{lem;quadratic interaction} and Theorem \ref{thm;free Bose field}.
\end{proof}

\appendix

\section{Operator theory and auxiliary results}\label{sect;appA}

In this appendix, we collect some results from operator theory and auxiliary results used in this paper.

\subsection{Operator theory}

\begin{lemma}\label{intertwining Borel}
	Let $S_1,S_2$ be self-adjoint operators acting in $\mathscr{H}$.
	Let $B$ be a bounded operator on $\mathscr{H}$.
	Suppose that $BS_1\subset S_2B$.
	Then, for any Borel function $f:\mathbb{R}\to\mathbb{C}$, we have $Bf(S_1)\subset f(S_2)B$.
\end{lemma}

\begin{proof}
	The assumption $BS_1\subset S_2B$ implies that $B(S_1-z)^{-1}=(S_2-z)^{-1}B$ for any $z\in\mathbb{C}\setminus\mathbb{R}$.
	Combining this with Stone's formula, we obtain $BE_{S_1}((a,b])=E_{S_2}((a,b])B$ for all $a,b\in\mathbb{R}$ with $a<b$, where $S_j=\int_\mathbb{R}\lambda\,\mathrm{d}E_{S_j}(\lambda)$ denotes the spectral resolution of $S_j$ for each $j=1,2$. 
	On the other hand, the set of all Borel subsets $J$ of $\mathbb{R}$ satisfying $BE_{S_1}(J)=E_{S_2}(J)B$ is a $\sigma$-algebra.
	Thus, $BE_{S_1}(J)=E_{S_2}(J)B$ holds for all Borel subsets $J$ of $\mathbb{R}$.
	The desired result now follows from the definition of $f(S_1)$ and $f(S_2)$.
\end{proof}

Recall that the sinc function is defined by
\begin{equation}\label{def;sinc}
\operatorname{sinc}x:=
\begin{dcases*}
	\frac{\sin x}{x}, & $x\in\mathbb{R}\setminus\{0\}$,\\
	1, & $x=0$.
\end{dcases*}
\end{equation}

\begin{lemma}\label{sinc eq}
	Let $S_1,S_2$ be non-negative self-adjoint operators acting in $\mathscr{H}$.
	Let $B$ be a bounded operator on $\mathscr{H}$.
	Suppose that
	\[
	B\operatorname{sinc}(tS_1)=\operatorname{sinc}(tS_2)B,\qquad \forall t\in\mathbb{R}.
	\] 
	Then, $BS_1\subset S_2B$ holds.
\end{lemma}

\begin{proof}
	For each $j=1,2,$ let $S_j=\int_0^\infty\lambda\,\mathrm{d}E_{S_j}(\lambda)$ be the spectral resolution of $S_j.$ 
	For any $f\in \operatorname{dom}(S_1^4),$ an elementary inequality
	\begin{equation*}
		-\frac{x^2}{120}\leq\frac{1-\operatorname{sinc}x}{x^2}-\frac{1}{6}\leq0,\qquad\forall x\in\mathbb{R}\setminus\{0\}
	\end{equation*}
	yields that
	\begin{align*}
		\left\|\frac{1-\operatorname{sinc}(tS_1)}{t^2}f-\frac{S_1^2}{6}f\right\|^2
		&=\int_0^\infty\left|\frac{1-\operatorname{sinc}(t\lambda)}{t^2}-\frac{\lambda^2}{6}\right|^2\,\mathrm{d}\|E_{S_1}(\lambda)f\|^2\\
		&\leq\int_0^\infty\left(\frac{t^2\lambda^4}{120}\right)^2\,\mathrm{d}\|E_{S_1}(\lambda)f\|^2\to0
	\end{align*}
	as $t\to0.$
	This, together with Fatou's lemma, implies that
	\begin{align*}
		&\int_0^\infty\left(\frac{\lambda^2}{6}\right)^2\,\mathrm{d}\|E_{S_2}(\lambda)Bf\|^2
		\leq\liminf_{t\to0}\int_0^\infty\left(\frac{1-\operatorname{sinc}(t\lambda)}{t^2}\right)^2\,\mathrm{d}\|E_{S_2}(\lambda)Bf\|^2\\
		&=\liminf_{t\to0}\left\|\frac{1-\operatorname{sinc}(tS_2)}{t^2}Bf\right\|^2
		=\liminf_{t\to0}\left\|B\frac{1-\operatorname{sinc}(tS_1)}{t^2}f\right\|^2
		=\left\|B\frac{S_1^2}{6}f\right\|^2<\infty.
	\end{align*}
	Thus, we have $Bf\in \operatorname{dom}(S_2^2)$ and $\|S_2^2Bf\|\leq\|BS_1^2f\|.$
	Since $\operatorname{dom}(S_1^4)$ is a core for $S_1^2,$ we obtain $B\operatorname{dom}(S_1^2)\subset \operatorname{dom}(S_2^2)$.
	Lebesgue's dominated convergence theorem, together with an elementary inequality
	\[
	0\leq\frac{1-\operatorname{sinc}x}{x^2}\leq\frac{1}{6},\qquad\forall x\in\mathbb{R}\setminus\{0\},
	\]
    yields that
	\begin{align*}
	\langle f,S_2^2Bg\rangle
	&=\int_0^\infty\lambda^2\,\mathrm{d}\langle f,E_{S_2}(\lambda)Bg\rangle
	=\lim_{t\to0}6\int_0^\infty\frac{1-\operatorname{sinc}(t\lambda)}{t^2}\,\mathrm{d}\langle f,E_{S_2}(\lambda)Bg\rangle\\
	&=\lim_{t\to0}6\left\langle f,\frac{1-\operatorname{sinc}(tS_2)}{t^2}Bg\right\rangle
	=\lim_{t\to0}6\left\langle f,B\frac{1-\operatorname{sinc}(tS_1)}{t^2}g\right\rangle\\
	&=\lim_{t\to0}6\int_0^\infty\frac{1-\operatorname{sinc}(t\lambda)}{t^2}\,\mathrm{d}\langle B^*f,E_{S_1}(\lambda)g\rangle
	=\langle B^*f,S_1^2g\rangle
	\end{align*}
for all $f\in\mathscr{H}$ and $g\in\operatorname{dom}(S_1^2)$.
Thus, we obtain $BS_1^2\subset S_2^2B$.
This, combined with Lemma~\ref{intertwining Borel}, implies $BS_1\subset S_2B$.
\end{proof}

\begin{lemma}\label{sin eq}
	Let $S_1,S_2$ be non-negative self-adjoint operators acting in $\mathscr{H}$.
	Let $B_+,B_-$ be bounded operators on $\mathscr{H}$.
	Suppose that
	\[
	B_-\sin(tS_1)=\sin(tS_2)B_+,\qquad \forall t\in\mathbb{R}.
	\] 
	Then, $B_-S_1\subset S_2B_+$ holds.
\end{lemma}

\begin{proof}
	For each $j=1,2,$ let $S_j=\int_0^\infty\lambda\,\mathrm{d}E_{S_j}(\lambda)$ be the spectral resolution of $S_j.$ 
	For any $f\in \operatorname{dom}(S_1^3),$ an elementary inequality
	\begin{equation*}
		-\frac{x^2}{6}\leq\frac{\sin{x}}{x}-1\leq0,\qquad\forall x\in\mathbb{R}\setminus\{0\}
	\end{equation*}
	yields that
	\begin{align*}
		\left\|\frac{\sin(tS_1)}{t}f-S_1f\right\|^2
		&=\int_0^\infty\left|\frac{\sin(t\lambda)}{t}-\lambda\right|^2\,\mathrm{d}\|E_{S_1}(\lambda)f\|^2\\
		&\leq\int_0^\infty\left(\frac{t^2\lambda^3}{6}\right)^2\,\mathrm{d}\|E_{S_1}(\lambda)f\|^2\to0
	\end{align*}
	as $t\to0.$
	This, together with Fatou's lemma, implies that
	\begin{align*}
		&\int_0^\infty\lambda^2\,\mathrm{d}\|E_{S_2}(\lambda)B_+f\|^2
		\leq\liminf_{t\to0}\int_0^\infty\left(\frac{\sin(t\lambda)}{t}\right)^2\,\mathrm{d}\|E_{S_2}(\lambda)B_+f\|^2\\
		&=\liminf_{t\to0}\left\|\frac{\sin(tS_2)}{t}B_+f\right\|^2
		=\liminf_{t\to0}\left\|B_-\frac{\sin(tS_1)}{t}f\right\|^2
		=\left\|B_-S_1f\right\|^2<\infty.
	\end{align*}
	Thus, we have $B_+f\in \operatorname{dom}(S_2)$ and $\|S_2B_+f\|\leq\|B_-S_1f\|.$
	Since $\operatorname{dom}(S_1^3)$ is a core for $S_1,$ we obtain $B_+\operatorname{dom}(S_1)\subset \operatorname{dom}(S_2)$.
	Lebesgue's dominated convergence theorem, together with an elementary inequality
	\[
	0\leq\frac{\sin{x}}{x}\leq1,\qquad\forall x\in\mathbb{R}\setminus\{0\}
	\]
	yields that
	\begin{align*}
		\langle f,S_2B_+g\rangle
		&=\int_0^\infty\lambda\,\mathrm{d}\langle f,E_{S_2}(\lambda)B_+g\rangle
		=\lim_{t\to0}\int_0^\infty\frac{\sin(t\lambda)}{t}\,\mathrm{d}\langle f,E_{S_2}(\lambda)B_+g\rangle\\
		&=\lim_{t\to0}\left\langle f,\frac{\sin(tS_2)}{t}B_+g\right\rangle
		=\lim_{t\to0}\left\langle f,B_-\frac{\sin(tS_1)}{t}g\right\rangle\\
		&=\lim_{t\to0}\int_0^\infty\frac{\sin(t\lambda)}{t}\,\mathrm{d}\langle B_-^*f,E_{S_1}(\lambda)g\rangle
		=\langle B_-^*f,S_1g\rangle
	\end{align*}
	for all $f\in\mathscr{H}$ and $g\in\operatorname{dom}(S_1)$.
	Hence, we obtain $B_-S_1\subset S_2B_+$.
\end{proof}

\subsection{Operator theory on Boson Fock space}

\begin{lemma}\label{lem;conti_of_field_op}
	Let $f\in\mathscr{H}$, and let $\{f_n\}_{n=1}^\infty$ be a sequence in $\mathscr{H}$ converging to $f$.
	Then we have
	\[
      \lim_{n\to\infty}\mathrm{e}^{\mathrm{i}\Phi_\mathrm{S}(f_n)}\Psi=\mathrm{e}^{\mathrm{i}\Phi_\mathrm{S}(f)}\Psi,\qquad\forall\Psi\in\mathscr{F}_\mathrm{b}(\mathscr{H}).
	\]
\end{lemma}

\begin{proof}
	See, e.g., \cite[Theorem 6.22 (iii)]{MR4292535} or \cite[Theorem 5.27 (iii)]{MR4812858}.
\end{proof}

\begin{lemma}\label{lem;expectation_value}
	It holds that
	\[
	\langle\Omega_\mathscr{H},\mathrm{e}^{\mathrm{i}\Phi_\mathrm{S}(f)}\Omega_\mathscr{H}\rangle=\mathrm{e}^{-\|f\|^2/4},\qquad\forall f\in\mathscr{H}.
	\]
\end{lemma}

\begin{proof}
	See, e.g., \cite[Theorem 6.18]{MR4292535} or \cite[Theorem 5.35]{MR4812858}.
\end{proof}

\begin{lemma}\label{lem;ccr_for_field_ops}
	For any $f,g\in\mathscr{H}$, we have the operator equality
	\[
	\mathrm{e}^{\mathrm{i}\Phi_\mathrm{S}(g)}\Phi_\mathrm{S}(f)\mathrm{e}^{-\mathrm{i}\Phi_\mathrm{S}(g)}
	=\Phi_\mathrm{S}(f)-\operatorname{Im}\langle g,f\rangle.
	\]
\end{lemma}

\begin{proof}
	See, e.g., \cite[Corollary 6.13]{MR4292535} or \cite[Corollary 5.12]{MR4812858}.
\end{proof}

\begin{lemma}\label{lem;trans_of_second_quant_op_by_field_op}
	Let $S$ be an injective non-negative self-adjoint operator acting in $\mathscr{H}$, and let $f\in\operatorname{dom}(S)$. 
	Then, we have the operator equality
	\[
	\mathrm{e}^{\mathrm{i}\Phi_\mathrm{S}(\mathrm{i}f)}\mathrm{d}\Gamma_\mathrm{b}(S)\mathrm{e}^{-\mathrm{i}\Phi_\mathrm{S}(\mathrm{i}f)}
	=H_S(Sf)+\frac{1}{2}\langle f,Sf\rangle.
    \]
\end{lemma}

\begin{proof}
	See, e.g., \cite[Lemma 13.5]{MR4812858}.
\end{proof}

\begin{lemma}\label{lem;trans_of_field_op_by_vH_hamiltonian}
	Let $S$ be an injective non-negative self-adjoint operator acting in $\mathscr{H}$, and let $g\in\operatorname{dom}(S^{-1/2})$. 
	Then, we have the operator equality
	\[
	\mathrm{e}^{\mathrm{i}tH_S(g)}\Phi_\mathrm{S}(f)\mathrm{e}^{-\mathrm{i}tH_S(g)}
	=\Phi_\mathrm{S}(\mathrm{e}^{\mathrm{i}tS}f)+\operatorname{Re}\left\langle S^{-1}(\mathrm{e}^{\mathrm{i}tS}-1)f,g\right\rangle,\qquad\forall t\in\mathbb{R},\ f\in\mathscr{H}.
	\]
\end{lemma}

\begin{proof}
	We first note that the operator $S^{-1}(\mathrm{e}^{\mathrm{i}tS}-1)$ in the right-hand side is bounded and defined on the whole $\mathscr{H}$.
	Then, the result follows from \cite[Theorem 13.6 (i)]{MR4812858} and a limiting argument in $f$.
\end{proof}

For a unitary operator $W$ on $\mathscr{H}$, we define a unitary operator $\Gamma_\mathrm{b}(W)$ on $\mathscr{F}_\mathrm{b}(\mathscr{H})$ by
\begin{equation}\label{eq;def_of_Gamma_b(W)}
\Gamma_\mathrm{b}(W):=\bigoplus_{n=0}^\infty \otimes^nW,
\end{equation}
where $\otimes^0 W:=1$.

\begin{lemma}\label{lem;trans_of_field_op_by_Gamma(U)}
	Let $W$ be a unitary operator on $\mathscr{H}$, and let $f\in\mathscr{H}$.
	Then, we have the operator equality
	\[
	\Gamma_\mathrm{b}(W)\Phi_\mathrm{S}(f)\Gamma_\mathrm{b}(W)^*=\Phi_\mathrm{S}(Wf).
	\]
\end{lemma}

\begin{proof}
	See, e.g., \cite[Theorem 6.7]{MR4292535} or \cite[Theorem 5.32]{MR4812858}.
\end{proof}

\subsection{Auxiliary results}

\begin{lemma}\label{lem;cauchy}
	Let $\{a_n\}_{n=1}^\infty$ be a sequence of real numbers satisfying
	\begin{equation}\label{eq;lem_of_cauchy}
	\lim_{m,n\to\infty}\mathrm{e}^{\mathrm{i}t(a_m-a_n)}=1,\qquad\forall t\in\mathbb{R}.
	\end{equation}
	Then, $\{a_n\}_{n=1}^\infty$ is a Cauchy sequence in $\mathbb{R}$. 
\end{lemma}

\begin{proof}
	We prove the lemma by contradiction.
	Suppose that $\{a_n\}_{n=1}^\infty$ is not a Cauchy sequence.
	Then, there exist $\varepsilon>0$ and strictly increasing sequences $\{m(k)\}_{k=1}^\infty$ and $\{n(k)\}_{k=1}^\infty$ of natural numbers such that $|a_{m(k)}-a_{n(k)}|\geq\varepsilon$ for all $k\in\mathbb{N}$.
	Set $x_k:=a_{m(k)}-a_{n(k)}$.
	By \eqref{eq;lem_of_cauchy}, we have $\lim_{k\to\infty}\cos(tx_k)=1$ for all $t\in\mathbb{R}$.
	This, together with Lebesgue's dominated convergence theorem, implies that
	\[
	\lim_{k\to\infty}\int_0^1\cos(tx_k)\,\mathrm{d}t=\int_0^1\lim_{k\to\infty}\cos(tx_k)\,\mathrm{d}t=1.
	\]
	Since $\int_0^1\cos(tx_k)\,\mathrm{d}t=\operatorname{sinc}(x_k)$, where the sinc function is defined by \eqref{def;sinc}, we obtain $\lim_{k\to\infty}\operatorname{sinc}(x_k)=1$.
	On the other hand, since $|x_k|\geq\varepsilon$ for all $k\in\mathbb{N}$, we have
	\[
	\operatorname{sinc}(x_k)\leq\sup_{|y|\geq\varepsilon}\operatorname{sinc}(y)<1,\qquad\forall k\in\mathbb{N},
	\]
	which contradicts $\lim_{k\to\infty}\operatorname{sinc}(x_k)=1$.
\end{proof}

\begin{lemma}\label{lem;conv}
	Let $\alpha\in\mathbb{R}$, and let $\{a_n\}_{n=1}^\infty$ be a sequence of real numbers.
	Suppose that
	\begin{equation*}
		\lim_{n\to\infty}\mathrm{e}^{\mathrm{i}t(a_n-\alpha)}=1,\qquad\forall t\in\mathbb{R}.
	\end{equation*}
	Then, $\{a_n\}_{n=1}^\infty$ converges to $\alpha$. 
\end{lemma}

\begin{proof}
	We have
	\[
	\mathrm{e}^{\mathrm{i}t(a_m-a_n)}=\mathrm{e}^{\mathrm{i}t(a_m-\alpha)}\times\mathrm{e}^{\mathrm{i}t(\alpha-a_n)}\xrightarrow{m,n\to\infty}1,\qquad\forall t\in\mathbb{R}.
	\]
	It follows from Lemma \ref{lem;cauchy} that $\{a_n\}_{n=1}^\infty$ is a Cauchy sequence.
	Let $\beta$ be a limit of $\{a_n\}_{n=1}^\infty$.
	Then, we obtain
	\[
	\mathrm{e}^{\mathrm{i}t(\beta-\alpha)}=\mathrm{e}^{\mathrm{i}t(\beta-a_n)}\times\mathrm{e}^{\mathrm{i}t(a_n-\alpha)}\xrightarrow{n\to\infty}1,\qquad\forall t\in\mathbb{R},
	\]
	which implies that $\alpha=\beta$.
\end{proof}

\section*{Acknowledgements}
I would like to thank Itaru Sasaki for carefully checking the proof of Theorem \ref{main thm} and for providing helpful comments.
This work was supported by JSPS KAKENHI Grant Numbers JP23K25783 and JP24K06755.

\section*{Data Availability}
No data were used to support this study.

\section*{Conflicts of interest}
The author declares that they have no conflicts of interest.

\section*{Use of AI}
During the preparation of this work, the author used ChatGPT (OpenAI) to improve the readability and language of the manuscript. 
The author reviewed and edited the content as needed and takes full responsibility for the content of the published article.

\bibliographystyle{plain}
\bibliography{References}

\end{document}